\documentclass[10pt, reqno]{article}
\usepackage{macros}
\newgeometry{margin=1in}
\usepackage{bibentry}
\usepackage[dvipsnames]{xcolor}
\usepackage{soul}

\newtheoremstyle{sltheorem} %
    {\topsep}                    %
    {\topsep}                    %
    {\slshape}                   %
    {}                           %
    {\bfseries}                   %
    {.}                          %
    {.5em}                       %
    {}  %

\theoremstyle{sltheorem}
\newtheorem{theorem}{Theorem}

\newtheorem{lemma}[theorem]{Lemma}

\newtheorem{as}{Assumption}

\theoremstyle{definition}

\title{Mostly Harmless Machine Learning: \\ Learning Optimal Instruments
in Linear IV Models\thanks{This work previously appeared in the
Machine Learning and Economic Policy Workshop at NeurIPS 2020. The authors
thank Isaiah Andrews, Mike Droste, Bryan
Graham, Jeff Gortmaker, Sendhil Mullainathan, Ashesh Rambachan, David
Ritzwoller, Brad Ross, Jonathan
Roth, Suproteem Sarkar, Neil Shephard, Rahul Singh, Jim Stock, Liyang Sun,
Vasilis
Syrgkanis, Chris Walker, Wilbur Townsend, and Elie Tamer for helpful comments.}}
\author{%
  Jiafeng Chen \\
  Harvard Business School\\
  Boston, MA \\
  \texttt{jchen@hbs.edu} \\
  \and
  Daniel L. Chen\\
  Toulouse School of Economics \\ 
  Toulouse, France \\ 
  \texttt{dlchen@nber.org}
  \and
  Greg Lewis\\ 
  Microsoft Research\\
  Cambridge, MA \\
  \texttt{glewis@microsoft.com}
}

\newcommand{\n}{N} %
\newcommand{\nh}{n} %
\newcommand{\yi}{Y_i} %
\newcommand{\di}{D_i} %
\newcommand{\covi}{X_i} %
\newcommand{\wi}{W_i} %
\newcommand{\ti}{T_i} %
\newcommand{\zi}{Z_i} %
\newcommand{\ui}{U_i} %
\newcommand{\instrument}{\tilde \upsilon} %
\newcommand{\Opt}{\Upsilon^\star} %
\newcommand{\limitopt}{\upsilon}
\newcommand{\limitOpt}{\Upsilon}

\renewcommand{\Pn}{\frac{1}{\n} \sum_{i=1}^\n}
\newcommand{\defeq}{\equiv}
\newcommand{\sample}{S}

\newcommand{\estopt}{\hat\limitopt}
\newcommand{\estOpt}{\hat\limitOpt}
\newcommand{\estoptj}{\hat\limitopt^{(j)}}
\newcommand{\estoptjtilde}{\tilde \limitopt^{(j)}}

\newcommand{\estOptj}{\hat\limitOpt^{(j)}}
\newcommand{\estOpti}{\hat\limitOpt_i}

\newcommand{\limitOpti}{\Upsilon_i}

\newcommand{\mlssn}{\hat\theta^{\mathrm{MLSS}}_{\n}}
\newcommand{\data}{R}

\newcommand{\G}{G} %
\newcommand{\VS}{\Omega} %
\newcommand{\skedas}{\sigma^2}

\newcommand{\condmean}{\mu}

\newcommand{\vi}{V_i} %
\newcommand{\vmte}{v} %
\newcommand{\MTE}{\operatorname{MTE}_a}

\newcommand{\weight}{w} %
\newcommand{\Epl}{\E_{\mathrm{(PL)}}}
\newcommand{\El}{\E_{\mathrm{(L)}}}
\newcommand{\err}{\mathrm{Err}} %
\newcommand{\AR}{\operatorname{AR}}

\newcommand{\tildeui}{\tilde U_i}
\newcommand{\vn}{V_{\nh,j}}
\newcommand{\omegan}{\Omega_{\nh,j}}
\newcommand{\tildexi}{\bar X_i}
\newcommand{\deltadiff}{(\tilde\delta - \delta)}

\begin{document}
\begin{titlepage}

\maketitle

\begin{abstract}
  We offer straightforward theoretical results that justify incorporating
  machine learning in the standard linear instrumental variable setting. The key
  idea is to use machine learning, combined with sample-splitting, to predict
  the treatment variable from the instrument and any exogenous covariates, and
  then use this predicted treatment and the covariates as technical instruments
  to recover the coefficients in the second-stage. This allows the researcher to
  extract non-linear co-variation between the treatment and instrument that may
  dramatically improve estimation precision and robustness by boosting
  instrument strength. Importantly, we constrain the machine-learned predictions
  to be linear in the exogenous covariates, thus avoiding spurious
  identification arising from non-linear relationships between the treatment and
  the covariates. We show that this approach delivers consistent and
  asymptotically normal estimates under weak conditions and that it may be
  adapted to be semiparametrically efficient \citep{chamberlain1992comment}. Our
  method preserves standard intuitions and interpretations of linear
  instrumental variable methods, including under weak identification, and
  provides a simple, user-friendly upgrade to the applied economics toolbox. We
  illustrate our method with an example in law and criminal justice, examining
  the causal effect of appellate court reversals on district court sentencing
  decisions.
\end{abstract}

\end{titlepage}

\section{Introduction}
\label{sec:intro}

Instrumental variable (IV) designs are a popular method in empirical economics.
Over 30\% of all NBER working papers and top journal publications considered by
\cite{currie2020technology} include some discussion of instrumental variables.
The vast majority of IV designs used in practice are linear IV estimated via
two-stage least squares (TSLS), a familiar technique in standard
introductions to econometrics and causal inference
\citep[e.g.][]{angrist2008mostly}. Standard TSLS, however, leaves on the table
some variation provided by the instruments that may improve precision of
estimates, as it only exploits variation that is linearly related to the
endogenous regressors. In the event that the instrument has a low linear
correlation with the endogenous variable, but nevertheless predicts the
endogenous variable well through a nonlinear transformation, we should expect
TSLS to perform poorly in terms of both estimation precision and inference
robustness. In particular, in some cases, TSLS would provide spuriously precise
but biased estimates \citep[due to weak instruments, see][]{andrews2018weak}.
Such nonlinear settings become increasingly plausible when exogenous variation
includes high dimensional data or alternative data, such as text, images, or
other complex attributes like weather. We show that off-the-shelf machine
learning techniques provide a general-purpose toolbox for leveraging such
complex variation, improving instrument strength and estimate quality.

Replacing the first stage linear regression with more flexible
specifications does not come without cost in terms of stronger identifying
assumptions. The validity of TSLS hinges only upon the restriction that the
instrument is linearly uncorrelated with unobserved disturbances in the response
variable. Relaxing the linearity requires that endogenous residuals are mean
zero conditional on the exogenous instruments, which is stronger. However, it is
rare that a researcher has a compelling reason to believe the weaker
non-correlation assumption, but rejects the slightly stronger mean-independence
assumption. Indeed, whenever researchers contemplate including higher
order polynomials of the instruments, they are implicitly accepting
stronger assumptions than TSLS allows.
In fact, by not exploiting the
nonlinearities, TSLS may accidentally
make a strong instrument weak, and deliver spuriously precise inference:
\cite{dieterle2016simple} and references therein find that several applied
microeconomics papers have conclusions that are sensitive to the specification
(linear vs. quadratic) of the first-stage.

A more serious identification concern with leveraging machine learning in the
first-stage comes from the parametric functional form in the second stage. When
there are exogenous covariates that are included in the parametric structural
specification, nonlinear transformations of these covariates could in principle
be valid instruments, and provide variation that precisely estimates the
parameter of interest. 
For example, in the standard IV setup of $Y = D^\t\tau + X^\t\beta
+ U$ where $X$ is an exogenous covariate, imposing $\E[U \mid X] = 0$ would
formally result in $X^2, X^3$, etc. being valid ``excluded'' instruments.
However, given that the researcher's stated source of identification comes
from excluded instruments, such ``identifying variation'' provided by
covariates is more of an artifact of parametric specification than any serious
information from the data that relates to the researcher's scientific inquiry.

One principled response to the above concern is to make the second stage
structural specification likewise nonparametric, thereby including an infinite
dimensional parameter to estimate, making the empirical design a
\emph{nonparametric instrumental variable} (NPIV) design. Significant
theoretical and computational progress have been made in this regard
\citep[inter alia,][]{newey2003instrumental,ai2003efficient,ai2007estimation,horowitz2007nonparametric,severini2012efficiency,ai2012semiparametric,hartford2017deep,dikkala2020minimax,chen2012estimation,chen2015sieve,chernozhukov2018double,chernozhukov2016locally}.
However, regrettably, NPIV has received relatively little attention in applied
work in economics, potentially due to theoretical complications, difficulty in
interpretation and troubleshooting, and computational scalability. Moreover, in
some cases parametric restrictions on structural functions come from theoretical
considerations or techniques like log-linearization, where estimated parameters
have intuitive theoretical interpretation and policy relevance. In these cases
the researcher may have compelling reasons to stick with parametric
specifications.

In the spirit of being user-friendly to practitioners, this paper considers
estimation and inference in an instrumental variable model where the second
stage structural relationship is linear, while allowing for as much nonlinearity
in the instrumental variable as possible, without creating unintended and
spurious identifying variation from included covariates. Our results
provide intuition and justification for
using machine learning methods in instrumental variable designs. We show that
with sample-splitting, under weak consistency conditions, a simple estimator
that uses the predicted values of endogenous and included regressors as
technical instruments is consistent, asymptotically normal, and
semiparametrically efficient. The constructed instrumental variable also readily
provides weak instrument diagnostics and robust procedures. Moreover, standard
diagnostics like out-of-sample prediction quality are directly related to
quality of estimates. In the presence of exogenous covariates that are
parametrically included in the second-stage structural function, adapting
machine learning techniques requires caution to avoid spurious identification
from functional forms of the included covariates. To that end, we formulate and
analyze the problem as a sequential moment restriction, and develop estimators
that utilize machine learning for extracting nonlinear variation from and only
from instruments.

\paragraph{Related Literature.} The core techniques that allow for the construction of our
estimators follow from \cite{chamberlain1987asymptotic,chamberlain1992comment}.
The ideas in our proofs are also familiar in the double machine learning
\citep{chernozhukov2018double,belloni2012sparse} and semiparametrics literatures
\citep[e.g.][]{liu2020deep}; our arguments, however, follow from elementary
techniques that are accessible to graduate students and are self-contained. Our
proposed estimator is similar to the split-sample IV or jackknife IV estimators
in \cite{angrist1999jackknife}, but we do not restrict ourselves to linear
settings or linear smoothers. Using nonlinear or machine learning in the first
stage of IV settings is considered by \cite{xu2021weakiv} (for probit),
\cite{hansen2014instrumental} (for ridge),
\cite{belloni2012sparse,chernozhukov2015post} (for lasso), and
\cite{bai2010instrumental} (for boosting), among others; and our work can be
viewed as providing a simple, unified analysis for practitioners, much in the
spirit of \cite{chernozhukov2018double}. To the best of our knowledge, we are
the first to formally explore practical complications of making the first stage
nonlinear in a context with exogenous covariates. Finally, we view our work as
counterpoint to the recent work by \cite{angrist2019machine}, which is more
pessimistic about combining machine learning with instrumental variables---a
point we explore in detail in \cref{sub:disc}.

\section{Main theoretical results}

We consider the standard cross-sectional setup where the data $(\data_i)_
{i=1}^\n = (\yi, \di, \covi, \wi)_ {i=1}^N \iid P$ are sampled from some
infinite population. $\yi$ is some outcome variable, $\di$ is a set of
endogenous treatment variables, $\covi$ is a set of exogenous controls, and
$\wi$ is a set of instrumental variables. The researcher is willing to argue
that $\wi$ is exogenously or quasi-experimentally assigned. Moreover, the
researcher believes that $\wi$ provides a source of variation that
``identifies'' an effect $\tau$ of $\di$ on $\yi$. We denote the endogenous
variables and covariates as $\ti\defeq  [1, \di^\t, \covi^\t]^\t$ and the
excluded instrument and covariates as the \emph{technical instruments} $\zi
\defeq [1, \wi^\t, \covi^\t]^\t$.

A typical specification in empirical economics is the linear instrumental
variables specification:
\begin{align}
\yi &= \alpha + \di^\t \tau + \covi^\t \beta + \ui \quad \E[\wi \ui] = 0.
\label{eq:linIV}
\end{align}
We believe that often the researcher is willing to assume more than that $\ui$
is uncorrelated with $(\covi, \wi)$. Common introductions of instrumental
variables \citep{angrist2008mostly,angrist2001instrumental} stress that
instruments induce variation in $\di$ and are otherwise unrelated to $\ui$, and
that a common source of instruments is natural experiments. We argue that these
narratives imply a stronger form of exogeneity than TSLS requires. After all, a
symmetric mean-zero random variable $S$ is uncorrelated with $S^2$, but one
would be hard pressed to say that $S^2$ is unrelated to $S$. Moreover, the condition $\E[\wi \ui] = 0$, strictly
speaking, does not automatically make polynomial expansions of $\wi$ valid
instruments, yet using higher order polynomials is common in empirical
research, indicating that the conditional restriction $\E[\ui \mid \wi] =
0$ more accurately captures the assumptions imposed in many empirical
projects.  With this in mind,
we will assume mean independence throughout the paper: $\E[\ui
\mid \wi] = 0$. 
This stronger exogeneity assumption allows researcher to extract more
identifying variation from instruments, but doing so calls for more flexible
machinery for dealing with the first stage.\footnote{Moreover, under
conditions such that our proposed estimator is
efficient, we can test mean independence
assuming $\E[\wi\ui] = 0$, since TSLS and our proposed
estimator are two
estimators that generate a Hausman test.}

\subsection{No covariates}
\label{sub:simple}
Let us first consider the case in which we do not have exogenous
covariates $\covi$. Our mean-independence restrictions give rise to a conditional
moment restriction, $
\E[\yi - \ti^\t \theta \mid W_i] = 0
$, where $\theta = (\alpha, \tau^\t, \beta^\t)^\t$. The conditional moment
restriction encodes an infinite set of unconditional
moment restrictions: \[
\text{For all square integrable $\instrument$}: \E[\instrument(\wi)(\yi - \ti^\t\theta)] = 0.
\]
\cite{chamberlain1987asymptotic} finds that all relevant statistical information
in a conditional moment restriction is contained in a single unconditional
moment restriction involving an \emph{optimal instrument} $\Opt$, and the
unconditional moment restriction with the optimal instrument delivers
semiparametrically efficient estimation and inference. In our case, $\Opt(\wi) =
\frac{1}{\skedas(\wi)}[1, \condmean(\wi)^\t]^\t$, where $
\condmean(\wi) \defeq \E[\di \mid \wi]$ and $\skedas(\wi) = \E[\ui^2 \mid
\wi]$. We estimate $\Opt$
 with $\estOpt$ and form a plug-in estimator for $\theta$:
\begin{equation}
    \hat\theta_\n = \pr{\Pn \estOpt(\wi) \ti^\t}^{-1} \pr{\Pn \estOpt(\wi)
    \yi}.
\label{eq:plugin}
\end{equation}
This is numerically equivalent to estimating \eqref{eq:linIV} with two-stage
weighted least-squares with $\hat \condmean(\wi)$ as an instrument and weighting
with $1/\skedas(\wi)$. In particular, if $\ui$ is homoskedastic, the
optimal
instrument is simply $[1, \mu(\wi)^\t]^\t$, and two-stage least-squares with an
estimate of $\mu(\wi)$ returns an estimate
$\hat\theta_\n$.\footnote{This approach should not be confused with what many applied researchers think of when they think of two-stage least squares, namely directly regressing $\yi$
on the estimated instrument $\estOpt(\wi)$ by OLS - i.e.     $\hat\theta_\n = \pr{\Pn \estOpt(\wi) \estOpt(\wi)^\t}^{-1} \pr{\Pn \estOpt(\wi)
    \yi}.$  This is what \cite{angrist2008mostly}
term the ``forbidden regression'', and it will not generally return
consistent estimates of $\theta$.} Under
heteroskedasticity, this instrument is no longer optimal (in the sense of
semiparametric efficiency) but remains valid. Therefore, we shall refer to the
instrument with the weighting $1/\skedas (\wi)$ as the \emph{optimal instrument
under efficient weighting} and the instrument without $1/\skedas(\wi)$ as the
\emph{optimal instrument under identity weighting}.%

Under identity-weighting, estimating $\Opt$ amounts to learning $\condmean
(\wi) \defeq \E[\di \mid
\wi]$, which is well-suited
to machine learning techniques; this is only slightly complicated by the
estimation of $\skedas(\wi) \defeq \E[\ui^2 \mid \wi]$ under efficient
weighting. One might worry that the preliminary estimation of $\Opt$ complicates
asymptotic analysis of $\hat \theta_\n$. Under a simple sampling-splitting
scheme, however, we state a high-level condition for consistency, normality, and
efficiency of $\hat \theta_\n$. Though it simplifies the proof and potentially
weakens regularity conditions, sample-splitting does reduce the amount of data
used to estimate the optimal instrument $\Opt$, but such problems can be
effectively mitigated by $k$-fold sample-splitting: 20-fold sample-splitting,
for instance, limits the loss of data to $5\%$ at the cost of 20 computations
that can be effectively parallelized. Such concerns notwithstanding, we focus
our exposition to two-fold sample-splitting.

Specifically, assume $\n=2\nh$ for simplicity and let $\sample_1, \sample_2
\subset [\n]$ be the two subsamples with size $\nh$. Under
identity-weighting, for $j \in \{1,2\}$, form
$\estOptj$ by estimating $\condmean(\wi)$ with data from the
other sample, $\sample_{-j}$. An estimator for $\condmean$ may be a neural
network or a random forest trained via empirical risk minimization, or a
penalized linear regression such as elastic net.\footnote{With $k$-fold
sample-splitting, $S_{-j}$ is the union of all sample-split folds other than the
$j$-th one.} The estimated instrument $\estOpti$ is
then formed by evaluating $\estOptj (\wi)$ for all $ i
\in S_j$. We may then use \eqref{eq:plugin} to form an (identity-weighted)
estimator of $\theta$ by plugging in $\estOpt$. Under efficient weighting, on
each $\sample_{-j}$, we would use the identity-weighted estimator of $\theta$ as
an initial estimator to obtain an estimate of $\ui$, and similarly predict
$\ui^2$ with $\wi$ to form an estimate of $\skedas(\wi)$. The resulting
estimated optimal instrument under efficient weighting may then be plugged into
\eqref{eq:plugin} and form an efficient-weighted estimator. We term such
estimators the machine learning split-sample (MLSS) estimators. The pseudocode for
major procedures considered in this paper is collected in \cref{algo}.

\Cref{thm:simple} shows that the MLSS estimator is consistent and asymptotically
normal when the first-stage estimator $\estOptj$ converges to  a strong
instrument. Moreover, it is semiparametrically efficient when $\estOptj$ is
consistent for the optimal instrument $\Opt(\wi) \defeq [1, \condmean
(\wi)^\t]/\skedas(\wi)$ in $L^2(W)$ norm. The $L^2$ consistency
condition\footnote{In many-instrument settings under
\cite{bekker1994alternative}-type asymptotic sequences, there may be no
consistent estimator of the optimal instrument in the absence of sparsity
assumptions \citep{raskutti2011minimax}.} is not strong---in particular, it is
weaker than the $L^2$ consistency at $o(N^ {-1/4})$-rate conditions commonly
required in the  double machine learning and semiparametrics literatures
\citep{chernozhukov2018double},\footnote{We are not claiming that the MLSS
procedure has any advantage over the double machine learning literature, but
simply that the statistical problem here is sufficiently well-behaved such that
we enjoy weaker conditions than is typically required.} where such conditions
are considered mild.\footnote{The nuisance parameter $\E[\di \mid \wi]$ in this
setting enjoys higher-order orthogonality property described in
\cite{mackey2018orthogonal}. In particular, it is infinite-order orthogonal,
thereby requiring no rate condition to work. Intuitively, estimation error in
$\limitOpt(\cdot)$ has no effect on the moment condition $\E[\limitOpt (\yi -
\alpha - \ti^\t \theta)] = 0$ holding, and this feature of the problem makes the
estimation robust to estimation of $\limitOpt$.}

Formally, regularity conditions are stated in \cref{cond:simple}. The first
condition simply states that the nuisance estimation attains \emph{some} limit
as sample size tends to infinity, which is a similar requirement as \emph{sampling
stability} in \cite{lei2018distribution}. The second condition states that the
limit is a strong instrument. The third condition assumes bounded moments so as
to ensure a central limit theorem. The last condition, which is only required
for semiparametric efficiency, states that the nuisance estimation is consistent
for the optimal instrument in $L^2$ norm. For consistency of standard error
estimates, we assume more bounded moments in \cref{as:variance}.
\begin{as}
\label{cond:simple}
Recall that $\zi = [1, \wi^\t]$, and so $\limitOpt(\wi)$ and $\limitOpt
(\zi)$
denote the same object.
\begin{enumerate}
    \item ($\estOptj$ attains a limit $\limitOpt$ in $L^2$ distance) There
    exists some measurable function $\limitOpt(\zi)$ such
    that \[
    \E{\norm{\estOptj(\zi) - \limitOpt(\zi)}^2 } \to 0 \quad \text{ for $j
    = 1,2$,}
    \]
    where the expectation integrates over both the randomness in $\estOptj$
    and in $\zi$, but $\estOptj$ and $\zi$ are assumed to be
    independent. 
    
    \item (Strong identification) The matrix $\G \defeq \E[\limitOpt(\zi)
    \ti^\t]$ exists and is full
    rank.
     
    \item (Lyapunov condition) (i) For some $\epsilon > 0$, the following
    moments
    are finite:
    $\E|\ui|^{2+\epsilon} < \infty$, 
    $\E\norm{\ui \limitOpt(\zi)}^{2 + \epsilon} <
    \infty$, $\E[
  \norm{\ti}^2] < \infty$,
  and (ii) The variance-covariance matrix $\Omega \defeq \E[\ui^2 \limitOpt
  (\zi)\limitOpt(\zi)^\t]$
  exists, and (iii) the conditional variance is uniformly bounded: For some
  $M$, $\E[\ui^2 \mid \zi] < M < \infty$ a.s.
  \item (Consistency to the optimal instrument) We may
  take the optimal instrument $\Opt(\zi)$ as the limit $\limitOpt(\zi)$ in
  condition 1.
\end{enumerate}
\end{as}

\begin{as}[Variance estimation]
\label{as:variance}
Let $\limitOpt$ be the object defined in \cref{cond:simple}. Assume that the following fourth moments are bounded:
    \[\max\br{\E[\norm{\ti}^4], \E[\ui^4], \E\norm{\limitOpt(\zi)}^4,
    \limsup_
    {\n\to\infty} \E[\norm{\estOptj(\zi)} ^4]} <
    \infty.\]
\end{as}

\begin{theorem}
\label{thm:simple}
  Let $\mlssn$ be the MLSS estimator
  described above.  Under conditions 1--3 in \cref{cond:simple}, \[
 \sqrt{\n} \pr{\mlssn - \theta} \wto
  \Norm(0, V) \quad\quad V \defeq (\G\VS^{-1}\G^\t)^{-1} = \G^{-1} \VS \G^
  {-\t},
  \]
  where $\G, \VS$ are defined in \cref{cond:simple}. Moreover, if condition
  4 in \cref{cond:simple} holds, then the asymptotic variance $V$ attains
  the semiparametric efficiency bound. Moreover, if we additionally assume 
  \cref{as:variance}, then the sample counterparts of $\G, \VS$ are
consistent for the two matrices. 
\end{theorem}

\begin{proof}[Proof of \cref{thm:simple}]
We may compute that the scaled estimation error is\[ 
    \sqrt{\n} \pr{\mlssn - \theta} = \pr{\Pn \hat\limitOpt(\zi)
    \ti^\t}^{-1}
    \frac{1}{\sqrt{\n}}\sum_{i=1}^\n \hat\limitOpt(\zi) U_i.
    \]
We verify in \cref{lemma:expansion} that, for $\limitOpt$ defined in
condition
1 of
\cref{cond:simple}, the following expansions hold:
\begin{equation}
    \label{eq:oracle}
    \hat \G \defeq \Pn \hat\limitOpt(\zi) \ti^\t = \Pn \limitOpt(\zi)
    \ti^\t + o_p(1) \qquad \frac{1}{
\sqrt{\n}}\sum_{i=1}^\n \hat\limitOpt(\zi) U_i = \frac{1}{\sqrt{\n}}\sum_ {i=1}^\n
\limitOpt(\zi) \ui + o_p(1).
\end{equation}
Expansion \eqref{eq:oracle} implies that $\mlssn$ is first-order equivalent
to the oracle estimator that plugs in $\limitOpt$: \[\hat \theta^\star_\n
\equiv \pr{\Pn \limitOpt(\zi)
    \ti^\t}^{-1}
    \Pn \limitOpt(\zi) Y_i,\]
     whose consistency and asymptotic normality follows from usual
     arguments under condition 3 of \cref{cond:simple}.
Given \eqref{eq:oracle}, then we have a law of large numbers $\Pn \hat\limitOpt(\zi)
\ti^\t \pto \G$ by condition
2 of \cref{cond:simple};
 and we obtain a central limit theorem $
\frac{1}{
\sqrt{\n}}\sum_
{i=1}^\n \limitOpt(\zi) \ui \wto \Norm(0, \Omega)$ by condition 3. Lastly, by
Slutsky's theorem and the fact that $\G$ is nonsingular, we obtain the
desired convergence $\smash{\sqrt{\n}
\pr{\mlssn
- \theta}  \wto \Norm(0,V)}$. 

If, additionally, we assume the consistency condition 4, then
$\hat\theta^\star_\n$ is exactly the efficient optimal instrument estimator
\citep{chamberlain1987asymptotic}, and hence $V$ attains the semiparametric
efficiency bound.
Finally, \eqref{eq:oracle} implies
that $\hat \G \pto \G$ via a weak law of large numbers, and
\cref{lemma:variance} implies $\hat \VS \defeq \Pn
(\yi - \ti^\t\mlssn)^2 \estOpti\estOpti^\t \pto \VS$, and thus the variance
can be consistently estimated.
\end{proof}

\subsection{Exogenous covariates}

The presence of covariates $\covi$ complicates the analysis considerably. Under
the researcher's model, both $\wi$ and $\covi$ are considered exogenous, and
thus we may assume $\E[\ui \mid \zi] = 0$ and use it as a conditional moment
restriction, under which the efficient instrument is $\var(\ui \mid \zi)^ {-1}
\E[\ti \mid \zi]$ and our analysis from the previous section continues to
apply \emph{mutatis mutandis}. However, if the researcher maintains a linear
specification $\yi = \ti^\t \theta + \ui$, estimating $\theta$ based on the
conditional moment restriction $\E[\ui \mid \zi] = 0$ may inadvertently
``identify'' $\theta$ through nonlinear behavior in $\covi$ rather than the
variation in $\wi$. Such a specification may allow the researcher to precisely
estimate $\theta$ even when the instrument $\wi$ is completely irrelevant, when,
say, higher-order polynomial terms in the scalar $\covi$, $\covi^2, \covi^3$,
are strongly correlated with $\di$, perhaps due to misspecification of the
linear moment condition. There may well be compelling reasons why these
nonlinear terms in $\covi$ allow for identification of $\tau$ under an economic
or causal model; however, they are likely not the researcher's stated source of
identification, and allowing their influence to leak into the estimation
procedure undermines credibility of the statistical exercise.

One idea to resolve such a conundrum is to make the structural function
nonparametric as well, and convert the model to a nonparametric instrumental
variable regression
\citep{newey2003instrumental,ai2003efficient,ai2007estimation,ai2012semiparametric,chen2012estimation}
(See \cref{asec:npiv} for  discussion).\footnote{More recently,
\cite{chernozhukov2018double} derive Neyman-orthogonal moment conditions
assuming a partially linear second stage in Section 4.2 of their paper.} Another
idea, which we undertake in this paper, is to weaken the  moment condition and
rule nonlinearities in $\covi$ as inadmissible for inference.

To that end, we analyze the statistical restrictions implied by the
model and consider relaxations. The conditional moment restriction $\E[\ui
\mid \zi] = 0$ is equivalent to the following orthgonality constraint
\begin{equation}
\label{eq:cmm}
  \text{For all (square integrable) $\limitOpt$},\, \E[\limitOpt(\wi, \covi)
(\yi-\ti^\t
\theta)] = 0.
\end{equation}
Condition \eqref{eq:cmm} is too strong, since it allows nonlinear transforms of
$\covi$ to be
valid instruments. A natural idea is to restrict the class of allowable
instruments $\limitOpt (\wi, \covi)$ to those that are partially linear in
$\covi$, $\limitOpt(\wi, \covi) = h(\wi) + \covi^\t\ell$, thereby deliberately
discarding information from nonlinear transformations of $\covi$. Doing so
yields the following family of orthogonality constraints:
\begin{equation}
    \text{For all (square integrable) $\limitOpt$},\, \E[\limitOpt(\wi)
(\yi-\ti^\t \theta)] = \E[\covi (\yi - \ti^\t \theta)] = 0.
\label{eq:restriction}
\end{equation}
We may view \eqref{eq:restriction} as imposing an orthogonality condition
on the structural errors $\ui$ that is intermediate between that of TSLS
and that of \eqref{eq:cmm}. In particular, if we
define $\Epl[\cdot \mid
\covi, \wi]$ as a projection operator that projects onto partially linear
functions of $(\covi, \wi)$: \[
\Epl[\ui \mid \covi, \wi]\defeq \argmin_{\substack{h(\covi, \wi) \\ h
(\covi,
\wi) =
\covi^\t \ell + g(\wi)}} \E\pr{\ui - h(\covi, \wi)}^2,
\]
then requiring \eqref{eq:restriction} is equivalent to requiring
orthogonality under this partially linear projection operator:
\begin{equation}
\Epl[\ui \mid \covi, \wi] = 0.
\label{eq:pl_restriction}
\end{equation} In contrast, the $\cov(\ui, \zi) = 0$ orthogonality
requirement of TSLS can be written as $\El [\cdot \mid \zi] = 0$, where
$\El[\cdot \mid \zi]$ is analogously defined as a projection operator onto
\emph{linear} functions of $\zi$. We see that \eqref{eq:pl_restriction} is a
natural interpolation between the respective orthogonality structures on the
errors $\ui$ induced by the TSLS and the conditional moment restrictions.

The moment restriction corresponding to
\eqref{eq:restriction} is the following \emph{sequential moment restriction}
\begin{equation}
  \E[\covi (\yi - \ti^\t \theta)] =  \E[\yi - \ti^\t \theta \mid \wi] = 0.
\label{eq:seqmom}
\end{equation}
We see that \eqref{eq:seqmom} is a natural interpolation between 
the usual
unconditional moment
condition, $\E[\zi \ui] =
0$, and the conditional moment restriction that may be spurious $\E[\ui
\mid \zi] =0$, by only allowing nonlinear information in $W_i$ to be
used for estimation and inference.

Having set up the estimation problem as (equivalently) characterized by 
\eqref{eq:restriction},
\eqref{eq:pl_restriction}, or \eqref{eq:seqmom}, efficient estimation is
discussed by 
\cite{chamberlain1992comment}. In particular, the optimal instrument under
identity weighting  takes the convenient
form 
\begin{equation}
    \Opt(\zi) = \Epl\bk{\ti \mid \covi, \wi} = \colvecb{3}{1}{\Epl 
\bk{\di \mid \covi, \wi}}{\covi}, 
\label{eq:opt_covariate_homo}
\end{equation}
which is simply $(1, \covi)$, along with the best \emph{partially
linear prediction} of the endogenous
treatment $\di$ from $\wi, \covi$. Observe that the only difference between
\eqref{eq:opt_covariate_homo} and
\cite{chamberlain1987asymptotic}'s optimal instrument under
homoskedasticity is modifying $\E$ into $\Epl$. 
Implementing 
\eqref{eq:opt_covariate_homo} is straightforward, as by 
\cite{robinson1988root},\footnote{Moreover, it is easy to impose partial
linear structure on certain estimators, including series regression and
feedforward neural
networks, and in those cases we may minimize squared error directly
without Robinson's transformation.} partially linear regression is
reducible
to
fitting two nonparametric regressions $\E[\di \mid \wi]$ and $\E[\covi
\mid \wi]$, and forming the following prediction function \[
\Epl\bk{\di \mid \covi, \wi} = \E[\di \mid \wi] + \El[\di \mid 
\pr{\covi - \E[\covi \mid \wi]}].
\]

\begin{table}[tb]
    \caption{List of nonparametric nuisance parameters that require
    estimation. Note that nuisance parameters that require the unobserved
    error $\ui$ require additional preliminary consistent estimators of
    $\theta$.}
    \label{tab:nuisance}
    \vspace{1em}
    \centering
    
    \begin{tabular}{ccc}
    \toprule 
        Covariates $\covi$ & Identity weighting? & Nonparametric
        nuisance
        parameters \\ 
        \midrule 
         No  & Yes & $\E[\di \mid \wi]$ \\ 
         No  & No & $\E[\di \mid \wi]$, $\E[\ui^2 \mid \wi]$ \\
         Yes & Yes & $\E[\di \mid \wi]$, $\E[\covi \mid \wi]$ \\
         Yes & No & $\E[\di \mid \wi]$, $\E[\covi \ui^2 \mid \wi]$, $\E
         [\ui^2 \mid \wi]$
         \\ 
        \bottomrule
    \end{tabular}
\end{table}

Efficient estimation in the heteroskedastic case is more complex. 
The optimal instrument is the vector \[
\Opt(\zi) = \frac{\E[\ti \mid \wi]}{\skedas(\wi)} + \E\bk{\ti \tilde
\covi^\t} \E
    \bk{\ui^2 \tilde \covi \tilde \covi^\t}^{-1}\tilde \covi  , \quad \tilde \covi \defeq \covi - \frac{\E[\covi \ui^2 \mid
    \wi]}
    {\skedas(\wi)}
\]
and the associated set of unconditional moment restrictions are \begin{equation}
    \E\bk{ \ui\cdot \Opt(\zi)
    } = 0.
    \label{eq:opt_cov_hetero}
\end{equation}
The intuition for \eqref{eq:opt_cov_hetero} is the following: the two
moment conditions $\E[\ui \covi] = \E[\ui \mid \wi] = 0$ provide
non-orthogonal information for $\theta$ that prevents us from applying the
optimal instrument on each moment condition. However, we may orthogonalize
one against the other.\footnote{Orthogonal here does not refer to Neyman orthogonality \citep{chernozhukov2018double}, but simply means that the two moments are uncorrelated.} In particular, the moment condition $\E[\tilde \covi
\ui] = 0$ is orthogonal to $\E[\ui \mid \wi]$ in the sense that $\E[\tilde \covi
\ui \cdot \ui \mid \wi] = 0$. Indeed, the term $\frac{\E[\covi \ui^2 \mid
    \wi]} {\skedas(\wi)}U_i \defeq \frac{\ip{\covi\ui, \ui}}{\ip{\ui,
\ui}} \ui $ is constructed to be the projection of $\covi \ui$
onto $\ui$ under the inner product $\ip{A,B} = \E[AB \mid \wi]$. 

As before, complications in nonparametric estimation can be avoided by
sample splitting, where nuisance parameters are estimated on $K-1$ folds of
the data and the moment condition is evaluated on the remaining fold.
As a summary across our settings, we collect the nuisance parameters that
require a first-step estimation in
\cref{tab:nuisance}. The estimator \[
\mlssn = \pr{\Pn \estOpt(\zi) \ti^\t }^{-1} \Pn \estOpt(\zi) \yi
\]
remains the same as \eqref{eq:plugin} and is subjected to the same
analysis in \cref{thm:simple}---under conditions 1--3 in 
\cref{cond:simple}, $\mlssn$
is consistent and asymptotically normal, and additionally, it is
semiparametrically efficient if the $L^2$-limit of $\estOpt$ coincides with
the optimal instrument $\Opt$ in their respective settings.

Lastly, we make two remarks about the case with exogenous
covariates, assuming identity weighting for tractability. First, it is possible
for $\E[\di \mid \wi] = 0$ and for \eqref{eq:opt_covariate_homo} to generate
precise estimates of the coefficient $\tau$. The reason is that it is possible
for the partially linear specification $\di = h(\wi) +
\covi^\t \ell + V_i$ to generate nonzero $h (\wi)$ but zero conditional
expectation, in much the same way that some regression coefficients may be zero
without adjusting for $\covi$, but nonzero when adjusted for $\covi$. Whether or
not this makes $\wi$ a plausibly exogenous and strong instrument is likely to be
context specific. A robustness check may be generated by replacing $\Epl[\di
\mid \covi, \wi]$ with $\E[\di \mid \wi]$, which delivers consistent and
asymptotically normal estimates (assuming strong instrument) at the cost of
efficiency. Second, a Frisch-Waugh-Lovell- or double machine learning-like 
procedure of first partialling out $\covi$ from $\yi,
\di$ and then treating \begin{equation}
    \yi - \El[\yi \mid \covi] = \tau\pr{\di - \El[\di \mid \covi]} + \ui \quad
\E[\ui \mid \wi] = 0
\label{eq:frisch-waugh}
\end{equation}
as a conditional moment restriction also delivers consistent and
asymptotically normal estimates. However, using the ``optimal'' instrument
for \eqref{eq:frisch-waugh}---the predicted residual $\E[D_i - \El[\di \mid
\covi] \mid \wi]$---\emph{does not} achieve semiparametric efficiency, since it
uses the information in the sequential moment restriction \eqref{eq:seqmom}
separately, without considering them jointly and orthogonalizing one against the
other, resulting in efficiency loss.

\begin{algorithm}[htb]
    \caption{Machine learning split-sample estimation and inference}
    \label{algo}
    \begin{algorithmic}
        \Require A subroutine PredictInstrument($\sample_{-j}, \sample_j$)
        that returns the estimated instrument $\{\estOpt(\zi) : i\in
        \sample_j\}$, where $\estOpt$ is
        a function of $\sample_{-j}$. 
        
         \vspace{1em}
        
        \Procedure{GenerateInstrument}{$K$, Data}
        \State Randomly split data into $\sample_1,\ldots,\sample_K$
        \For {$j$ in 1,\ldots, $K$} 
            \State $\estOptj \gets \mathrm{PredictInstrument}(\sample_{-j},
            \sample_j)$
        \EndFor
        \State Combine $\estOptj$ into $\estOpt$
        \State Return $\sample_1,\ldots,\sample_K, \estOpt$
        \EndProcedure
        
         \vspace{1em}
        
        \Procedure{MLSSEstimate}{$K$, Data}
        \State $\sample_1,\ldots,\sample_K, \estOpt \gets 
        \mathrm{GenerateInstrument}(K, \mathrm{Data})$
        
        \State For the full parameter vector, return \[\hat\theta = \pr{\Pn
        \estOpt(\wi) \ti^\t}^{-1} \pr{\Pn \estOpt(\wi)
    \yi}\] and variance estimate \[
    \hat V = \pr{\Pn \estOpti \ti^\t}\cdot \pr{\Pn (\yi - \ti^\t
    \hat\theta)^2 \estOpti \estOpti^\t}
    \cdot
    \pr{\Pn \estOpti \ti^\t}^{-\t}
    \]

        \State For the subvector $\hat\tau$, residualize $\estOpti, \yi,
        \di$ against $\covi$ to obtain $\tilde \upsilon_i, \tilde Y_i,
        \tilde D_i$, and compute $\hat\tau = \pr{\Pn \tilde \upsilon_i
        \tilde D_i^\t}^{-1} \pr{\Pn \tilde \upsilon_i\tilde Y_i}.$
        
        \Comment{\textbf{Assuming identity weighting}}
       
        \EndProcedure
        
        \vspace{1em}
        
        \Procedure{WeakIVInference}{$K$, Data, $\alpha$}
            
            \Comment{\textbf{Assuming identity weighting}}
        
            \Comment{Assuming a routine AndersonRubin($\alpha, 
            \mathrm{Data}$) that
            returns the $1-\alpha$ Anderson--Rubin CI}
            
            \State $\sample_1,\ldots,\sample_K, \estOpt \gets 
        \mathrm{GenerateInstrument}(K, \mathrm{Data})$
        
            \For {$j$ in 1,\ldots, $K$} 
            \State On $\sample_j$, residualize $\estOptj(\zi), \yi,
        \di$ against $\covi$ to obtain $\tilde \upsilon_i, \tilde Y_i,
        \tilde D_i$
            \State $\mathrm{CI}_j \gets \mathrm{AndersonRubin}\pr{\alpha /
            K,
            \br{\tilde \upsilon_i, \tilde Y_i,
        \tilde D_i}_{i\in \sample_j}}$
            \EndFor 
            \State Return $\mathrm{CI} = \bigcap_j \mathrm{CI}_j$
        \EndProcedure
    \end{algorithmic}
\end{algorithm}

\section{Discussion}
\label{sub:disc}

\textbf{``Forbidden regression.''} Nonlinearities in the first stage are often
discouraged due to a ``forbidden regression,'' where the researcher regresses
$Y$ on \smash{$\hat D$} estimated via nonlinear methods, motivated by a
heuristic explanation for TSLS. As \cite{angrist2001instrumental} point out,
this regression is inconsistent, and consistent estimation follows from using
$\hat D$ as an instrument for $D$, as we do, rather than replacing $D$ with
$\hat D$---in the case where the first-stage is linear, the two estimates are
numerically equivalent, but not in general.

\paragraph{Interpretation under heterogeneous treatment effects.}   
Assume $\di$ is binary and suppose $\yi = \di \yi(1) + (1-\di)\yi(0)$. Suppose
the treatment follows a Roy model, $\di = \one
\pr{\condmean(\wi) \ge \vi}$, where $\vi \sim \Unif[0,1]$. In this
setting, the conditional moment restriction \eqref{eq:linIV} is
misspecified, since it assumes constant treatment effects, and different
choices of the instrument would generate estimators that converge to
different population quantities. Nevertheless, the results of
\cite{heckman2005structural} (Section 4) show that different choices of the
instrument generate estimators that estimate different weightings of
\emph{marginal treatment effects} (MTEs); moreover, optimal instruments,
whether under identity weighting or optimal weighting, correspond to convex
averages of MTEs, whereas no such guarantees are available for linear IV
estimators with $\wi$ as the instrument, without assuming that $\E[\di \mid
\wi]$ is linear. The weights on the MTEs are explicitly stated in 
\cref{asec:hte}.

\paragraph{Weak IV detection and robust inference.} 

A major practical motivation for our work, following \cite{bai2010instrumental},
is to use machine learning to rescue otherwise weak instruments due to a lack of
linear correlation; nonetheless, the instrument may be irredeemably weak, and
providing weak-instrument robust inference is important in practice. Relatedly,
\cite{xu2021weakiv} and \cite{antoine2019identification} also consider weak IV
inference with nonlinear first-stages; the benefits of split-sampling in the
presence of many or weak instruments are recently exploited by
\cite{mikusheva2020inference} and date to \cite{dufour2003identification},
\cite{angrist1999jackknife}, \cite{staiger1994instrumental}, and references
therein; \cite{kaji2019theory} proposes a general theory of weak identification
in semiparametric settings.

On weak IV detection, our procedure produces estimated optimal instruments,
which result in just-identified moment conditions. As a result, in models with
\emph{a single endogenous treatment variable}, the \cite{stock2005testing}
$F$-statistic rule-of-thumb has its exact interpretation\footnote{Namely, the
worst-case bias of TSLS exceeds 10\% of the worst-case bias of OLS
\citep{andrews2018weak}.} regardless of homo- or heteroskedasticity
\citep{andrews2018weak}, and the first stage $F$-statistic may be used as a tool
for detecting weak instruments.

Pre-testing for weak instruments distorts downstream inferences. Alternatively,
weak IV \emph{robust} inferences, which are inferences of $\tau$ that are valid
regardless of instrument strength, are often preferred. The procedure we
propose, under identity-weighting, is readily compatible with simple robust
procedures. In particular, on each subsample $\sample_j$, we may perform the
Anderson--Rubin test \citep{anderson1949estimation} and combining the results
across subsamples via Bonferroni correction. For a confidence interval at the
95\% nominal level with two-fold sample-splitting, this amounts to intersecting
two 97.5\%-nominal AR confidence intervals on each subsample, and these
confidence intervals may be computed using off-the-shelf software
implementations of AR confidence intervals.

More formally, consider the null hypothesis $H_0: \tau = \tau_0$. Consider
a Frisch--Waugh--Lovell procedure that partials out the covariates $\covi$.
Let the residuals be $\ui
(\tau_0)
\defeq \yi
-\di^\t \tau_0$, and let $\tildeui(\tau_0) \defeq \ui(\tau_0) -
 \tilde\delta^\t \tildexi$ be the residual $\ui
(\tau_0)$ after partialling out $\tildexi \defeq [1, \covi^\t]^\t$. Suppose the
estimated instrument takes the form $\estOpt(\zi) = [1,
\estopt(\zi)^\t,
\covi^\t]^\t$, where $\dim\estopt(\zi) = \dim \di$; this requirement is satisfied under
identity weighting. Similarly, partial out $\tildexi$ from the estimated
instrument to obtain $\estoptjtilde
\defeq \estoptj - \tilde \lambda^\t \tildexi$. 
Finally, consider the covariance $\vn
(\tau_0)$ between
the residual and the instrument, after partialling out the covariates
$\tildexi$, and let $\omegan$ be an estimate of $\vn$'s variance matrix:
i.e.,
\[
\vn(\tau_0) \defeq \frac{1}{\sqrt{\nh}} \sum_{i\in\sample_j}
\estoptj
(\zi)\tildeui(\tau_0) \text{ and } \omegan(\tau_0) \defeq \frac{1}
{\nh}\sum_{i\in \sample_j} \tildeui (\tau_0)^{2}
\estoptjtilde (\zi)\estoptjtilde (\zi)^\t.
\] 
The Anderson--Rubin
statistic on the $j$-th subsample is then defined as the normalized
magnitude of $\vn$: $\AR_j(\tau_0) \defeq \vn^\t
\omegan^{-1} \vn.$ Under $H_0$, by virtue of the exclusion restriction, we
should expect $\vn$ be mean-zero Gaussian, and thus $\AR_j$ should be
$\chi^2$.
Indeed, \Cref{thm:weakiv} shows that on each subsample,
under mild bounded moment conditions that ensure convergence
(\cref{as:weakiv}), $\AR_j(\tau_0)$ attains a limiting $\chi^2$
distribution. Under weak IV asymptotics, it is not necessarily the case that
the AR statistics are asymptotically uncorrelated across subsamples, and so
we resort to the Bonferroni procedure in outputting a single confidence
interval.
\begin{as}[Bounded moments for the AR statistic]
    \label{as:weakiv}
    Without loss of generality and normalizing if necessary, assume the
    estimated instruments are normalized: $\sum_
    {i\in\sample_j}\estoptj_k
    (\zi)^2 = 1$ for all $k = 1, \ldots,\dim \di$. 
    Let $\lambda_n \defeq \E[\tildexi
\tildexi]^{-1} \E
[\estoptj(\zi) \tildexi^\t \mid \estoptj]$ be the projection coefficient of
$\estoptj(\zi)$ onto $\tildexi$. 
Assume that with probability 1, the sequence $\estoptj =
    \estoptj_\nh$ satisfies the Lyapunov conditions
    \begin{enumerate}[label=(\roman*)]
        \item $\E[\ui^{4} \norm{\estoptj(\zi) -
    \lambda_n \tildexi}^
    {4} \mid \estoptj] < C_1 < \infty$ for some $C_1 > 0$
    \item $\E[\ui^2 (\estoptj(\zi) - \lambda_n \tildexi) (\estoptj
        (\zi)- \lambda_n \tildexi)^\t \mid \estoptj] \pto
    \Omega$. 
    \end{enumerate}
    Moreover, assume (iii) $\max\br{\E[\norm{\estoptj}^4], \E[\ui^4], \E
    [\norm{\covi}^4]} < C_2 < \infty$ and that (iv) $\E
    [\tildexi\tildexi^\t]$
    is invertible.
\end{as}
\begin{theorem}
\label{thm:weakiv}
    Under \cref{as:weakiv}, $\AR_j(\tau_0) \wto \chi^2_{\dim
    \di}$.
\end{theorem}
\begin{proof}
    We relegate the proof, which amounts to checking convergences $\vn
    \wto \Norm(0,\Omega)$ and $\omegan \pto \Omega$ under \cref{as:weakiv}, to
    the appendix.
\end{proof}

\paragraph{Connection between first-stage fitting and estimate quality.}
By using machine learning in the first stage, one may be able to improve the quality of the first-stage fit, as measured by out-of-sample $R^2$.
We now offer an argument as to why improving that fit may improve the mean
squared error of the estimator.

Consider a setting with no covariates\footnote{In the presence of covariates
under identity weighting, we may, without loss of generality, partial out the
covariates after estimating the optimal instrument.} and i.i.d. $\ui$ and an
estimated instrument $\estopt(\zi)$,\footnote{i.e. $\estOpt(\zi) = [1,
\estopt(\zi)]^\t.$} meant to approximate $\E[\di \mid \zi]$. In linear IV, $\estopt(\zi)$ is the linear projection of
$\di$ onto $\zi$. Define the \emph{extra-sample error}\footnote{See, for
instance, \cite{friedman2001elements} Section 7.4.} of an estimator $\hat\tau$
based on $\estopt (\zi)$ to be the random quantity \begin{align*}
\err(\estopt) &\defeq { n \cdot \pr{\frac{\cov_n (\estopt(\zi) , \yi)}{ \cov_n
   (\estopt(\zi),
   \di)} - \tau}^2 }
\end{align*} where $(\yi,\di,\zi)_ {i=1}^n$ is a new and independent sample
unrelated to the estimate $\estopt$, and we hold $\estopt$ fixed. The subscript
$n$ denotes sample quantities such as sample variances and covariances. The
quantity $\err (\estopt)$ is an optimistic measure of the quality of using
$\estopt$ as the instrument, as construction of $\estopt$ without
sample-splitting typically introduces a bias term since
$\cov(\estopt, \ui)$ cannot be assumed to be small. The following
calculation shows
that  $\err(\estopt)$ scales with the inverse out-of-sample $R^2$ of $\estopt$
as a predictor of $\di$:
\begin{align*}
\err(\estopt) &= { n \frac{\cov_n (\estopt(\zi) , \ui)^2}{ \cov_n (\estopt(\zi),
   \di)^2}
} \\ &= {
    \frac{1}{R^2_{n}\pr{\estopt(\zi), \di} } \cdot \var_n (\di)^{-1} \pr{
    \frac{1}
    {\sqrt{n}}\sum_{i=1}^n \frac{\estopt(\zi)}{\sqrt{\var_n (\estopt
    (\zi))}} (\ui - \bar U)}^2 
}\\
&\wto \frac{1}{R^2\pr{\estopt(\zi), \di}}\frac{\var(\ui)}{\var(\di)}
\cdot \chi^2_1 \quad \text{as $n \to \infty$},
\end{align*}
where $R^2_n$ is the out-of-sample $R^2$ of predicting $\di$ with
$\estopt(\zi)$, and $R^2(\cdot,\cdot)$ is its limit in
probability.\footnote{The $F$-statistic is a monotone transformation of
the $R^2$, which also serves as an indication of estimation quality.}

The out-of-sample $R^2$, which can be readily computed from a split-sample
procedure, therefore offers a useful indicator for quality of estimation. In
particular, if one is comfortable with the strengthened identification
assumptions, there is little reason not to use the model that achieves the best
out-of-sample prediction performance on the split-sample. In some settings, this
best-performing model will be linear regression, but in many settings it may
not be, and attempting more complex tools may deliver considerable
benefits.

Moreover, much of the discussion on using machine learning for
instrumental variables analysis focuses on \emph{selecting} (relevant or
valid) instruments
\citep{belloni2012sparse,angrist2019machine} assuming some level of sparsity,
motivated by statistical difficulties encountered when the number of instruments
is high.  In light of the heuristic above, a more precise framing is perhaps
\emph{combining} instruments to form a better prediction of the endogenous
regressor, as noted by \cite{hansen2014instrumental}.

\paragraph{(When) is machine learning useful?} We conclude this section by
discussing our work relative to \cite{angrist2019machine}, who note that using
lasso and random forest methods in the first stage does not seem to provide
large performance benefits in practice, on a simulation design based on the data
of \cite{angrist1991does}. We note that, per our discussion above in the
connection between first-stage fitting and estimate quality, a good heuristic
summary for the estimation precision is the $R^2$ between the fitted instrument
and the true optimal instrument---$\E[\di\mid \wi]$ in the homoskedastic case.
It is quite possible that in some settings, the conditional expectation $\E [\di
\mid \wi]$ is estimated well with linear regression, and lasso or random forest
do not provide large benefits in terms of out-of-sample prediction quality.
Since \cite{angrist1991does}'s instruments are quarter-of-birth interactions and
are hence binary, it is in fact likely that predicting $D$ with  linear
regression performs well relative to nonlinear or complex
methods\footnote{Indeed, in some of our  experiments calibrated to the
\cite{angrist1991does} data, using a simple gradient boosting method
(\texttt{lightgbm}) does not outperform linear regression in terms of
out-of-sample $R^2$ (0.039\% vs. 0.05\%).} in the setting. Whether or not
machine learning methods
work well relative to linear methods is something that the researcher may verify
in practice, via evaluating performance on a hold-out set, which is standard
machine learning practice but is not yet widely adopted in empirical economics.
Indeed, we observe that in both real (\cref{sec:emp}) and Monte Carlo
(\cref{asec:monte_carlo}) settings where the out-of-sample prediction quality of
more complex machine learning methods out-perform linear regression, MLSS
estimators perform better than TSLS.

\section{Empirical Application}
\label{sec:emp}

We consider an empirical application in the criminal justice setting of
\cite{deepivinlaw}, where we consider the causal effect of appellate court
decisions at the U.S. circuit court level on lengths of criminal sentences at
the U.S. district courts under the jurisdiction of the circuit court.
\cite{deepivinlaw} exploit the fact that appellate judges are randomly assigned,
and use the characteristics of appellate judges---including age, party
affiliation, education, and career backgrounds---as instrumental variables. In
criminal justice cases, plaintiffs rarely appeal, as it would involve trying the
defendant twice for the same offense---generally not permitted in the United
States; therefore, an appellate court reversal would
typically be in favor of defendants, and we may posit a causal channel in which
such reversals affect sentencing; for instance, the district court may be more
lenient as a result of a reversal, as would be naturally
predicted if the reversal sets up a precedent in favor of the defendant.

To connect the empirical setting with our notation, the outcome variable $Y$ is
the change in sentencing length before and after an appellate decision, measured
in months, where positive values of $Y$ indicates that sentences become longer
after the appellate court decision. The endogenous treatment variable $D$ is
whether or not an appellate decision reverses a district court ruling. The
instruments $W$ are the characteristics of the randomly assigned circuit judge
presiding over the appellate case in question, and covariates $X$ contain
textual features from the circuit case, represented by Doc2Vec embeddings
\citep{le2014distributed}.

We compute two estimators of the optimal instrument under identity weighting
based on flexible methods that are often characterized as machine learning
(random forest and LightGBM, for light gradient boosting machine), as well as a
variety of linear or polynomial regression estimators, with or without
sample-splitting. We present our results in \cref{fig:judge} (see the notes to the figure for precise definitions of each estimator). For all of the split-sample estimators we have three sets of point estimates and confidence intervals, corresponding to the different splits.  We also have specifications that exclude (panel (a)) and include covariates (panel (b)); except where discussed below the results are consistent across both.     

The machine learning estimators perform similarly across splits, reporting
mildly negative point estimates
of between $1$ and $2$ months reduction in sentencing, with confidence intervals that are reasonably tight but include a zero effect. Moreover, the Wald and Anderson--Rubin
confidence intervals are similar, suggesting that the instrument constructed by the ML methods is sufficiently strong
to result in inferences that are not distorted. 

A natural benchmark to compare these results to is TSLS, with the instruments
entering either linearly or quadratically (but not interacted). Linear TSLS
estimates a slightly positive effect without covariates and a slightly negative
one with them. Though it has a tight Wald confidence interval, the AR interval
is quite large ($[-3.3, 23.8]$). This indicates a weak instruments
problem, and indeed the first-stage $F$-statistic is only 1.5. 
Quadratic TSLS returns
a point estimate that is close to the MLSS estimates, with a very tight Wald
confidence interval; however, the corresponding Anderson--Rubin interval is
empty. Anderson--Rubin tests
test model misspecification jointly with point nulls of the structural
coefficient, and may report an empty interval if the model is
misspecified.\footnote{In our case, the endogenous treatment is binary,
and so the only source of model misspecification is heterogeneous treatment
effects. In that case, TSLS continues to estimate population objects
that are (possibly nonconvex) averages of marginal treatment effects, and
arguably researchers would nonetheless like non-empty confidence sets. One
benefit of split-sample approaches is that the power of the
Anderson--Rubin test is wholly directed to testing the structural parameter
rather than testing overidentification, since the estimated instrument always
results in a just-identified system. As a result, Anderson--Rubin intervals
under split-sample approaches will never be empty. Another benefit of our
approach, which
we discuss in \cref{sub:disc} and \cref{asec:hte}, is that MLSS
consistently estimates a convex average of marginal treatment effects
assuming the first-stage is consistent for the conditional mean of the
endogenous treatment on the exogenous instruments.}
Moreover, we should still be concerned with weak instruments: The
first-stage $F$-statistic is only $2.2$
(excluding covariates) and $2.3$ (including covariates).

The sample-splitting estimators based on traditional polynomial expansions rather than machine learning all perform poorly, with out-of-sample $R^2$ close to zero and consequently huge confidence intervals (the point estimates also vary wildly across splits). Overall, the MLSS estimators successfully extract more variation from the instruments than the alternatives, and consequently deliver more statistical precision.

\begin{figure}[tb]
  {\centering
    (a) Excluding covariates
    \includegraphics[width=\textwidth]{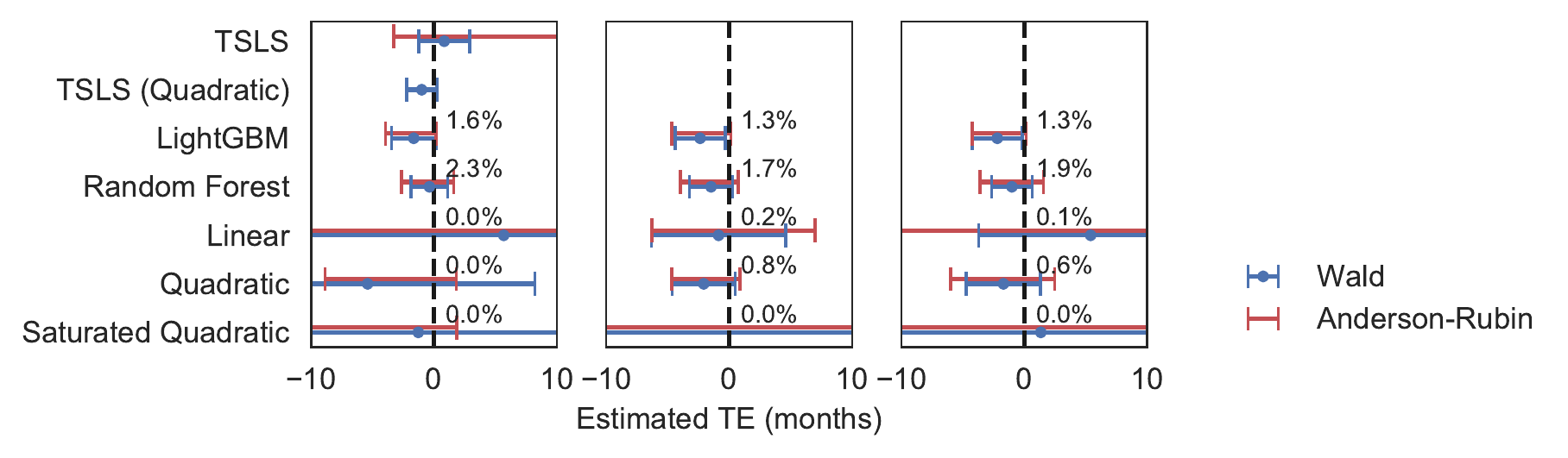}

    \vspace{1em}
    (b) Including covariates
    
    \includegraphics[width=\textwidth]{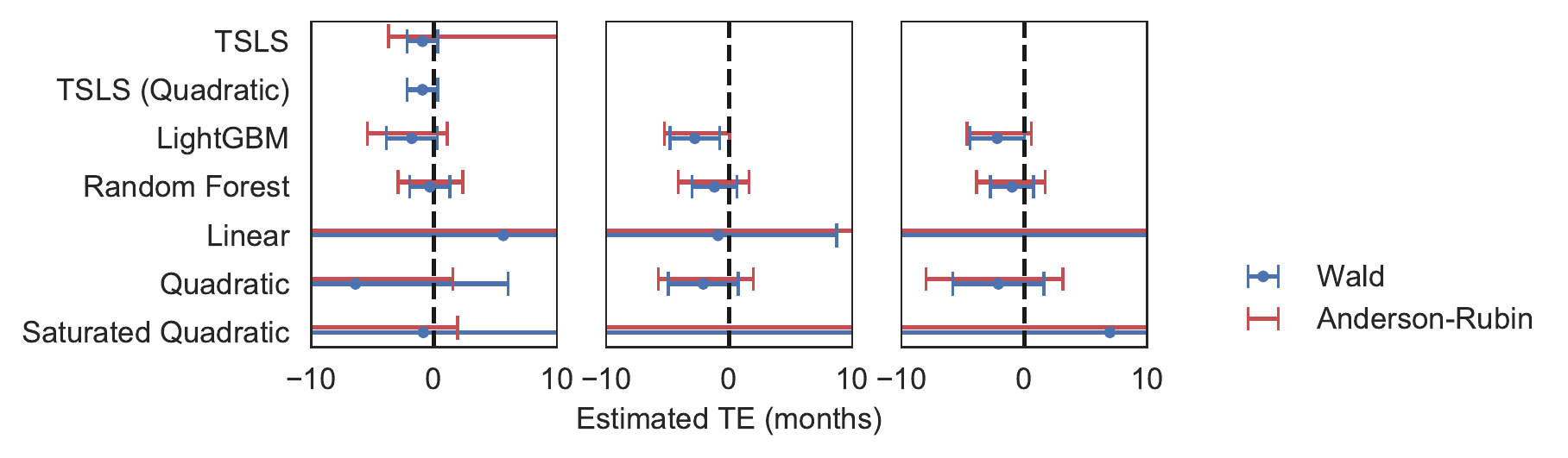}}
  
  Notes: Point estimate and confidence intervals across three sample
  splits, represented by the three horizontal panels. TSLS and TSLS 
  (Quadratic) are direct estimates without sample-splitting. Out-of-sample
  $R^2$ of instruments on endogenous treatment in annotation in panel (a).
  
  \quad \textsf{TSLS} is a standard TSLS estimator without sample
  splitting,
  using the instruments directly from the dataset.
  \textsf{TSLS (Quadratic)} includes second-order terms (but not
  interactions) for the instruments, again without sample-splitting---in
  particular, it results in an empty Anderson--Rubin interval.
  \textsf{LightGBM} and \textsf{RandomForest} are MLSS estimators, where
  LightGBM is an algorithm for gradient boosted trees. Finally, 
  \textsf{Linear}, \textsf{Quadratic}, and \textsf{Saturated Quadratic} are
  split-sample estimators with linear regression, quadratic regression 
  (without interactions), and quadratic regression with interactions as the estimators for the
  instrument, respectively.

  \caption{IV estimation of the effect of appellate court
  reversal on district court sentencing decisions}
  \label{fig:judge}
\end{figure}

\section{Conclusion}
In this paper, we provide a simple and user-friendly analysis of incorporating
flexible prediction into instrumental variable analysis in a manner that
is familiar to applied researchers. In particular, we document via elementary
techniques that a split-sample IV estimator with machine learning methods as the
first stage inherits classical asymptotic and optimality properties of usual
instrumental regression, requiring only weak conditions governing the
consistency of the first stage prediction. In the presence of covariates, we
also formalize moment conditions for instrumental regression that continues to
leverage nonlinearities in the excluded instrument without creating spurious
identification from the nonlinearities in the included covariates. Leveraging
such nonlinearity in the first stage allows the user to extract more identifying
variation from the instrumental variables and can have the potential of rescuing
seemingly weak instruments into strong ones, as we demonstrate with simulated
data and real data from a criminal justice context. Conventional components
of an instrumental variable analysis, such as identification-robust confidence
sets, extend seamlessly in the presence of a machine learning first stage. We
believe that machine learning in IV settings is a mostly harmless addition to
the empiricist's toolbox.

\newpage

\bibliographystyle{ecca}
\bibliography{main.bib}

\newpage
\appendix
\section{Technical lemmas and proofs}

\small

\begin{lemma}
\label{lemma:expansion}
    Under conditions 1 and 3 of \cref{cond:simple}, we have 
    \eqref{eq:oracle}.
\end{lemma}
\begin{proof}
    We first consider the first statement. Observe that \[
\Pn \pr{\estOpt(\zi) - \limitOpt(\zi)}\ti^\t = \frac{1}{2} \sum_{j\in
\{1,2\}} \frac{1}{\nh} \sum_{i\in
\sample_j} \pr{\estOptj(\zi) - \Upsilon(\zi)}\ti^\t 
\]
We control the right-hand side, where $\norm{\cdot}_F$ is the Frobenius
norm:
\begin{align*}
    \norm[\bigg]{\frac{1}{\nh} \sum_{i\in
\sample_j} \pr{\estOptj(\zi) - \Upsilon(\zi)}\ti^\t}_F^2 &\le \pr{\frac{1}{\nh}\sum_{i\in \sample_j} 
\norm{\estOptj(\zi) - \Upsilon(\zi)} \cdot \norm{\ti}}^{2} \tag{$
\norm{AB}_F \le
\norm{A}_F \norm{B}_F$} \\ 
&\le \pr{ \frac{1}{\nh}\sum_{i\in \sample_j} \norm{\estOptj(\zi) - \Upsilon(\zi)}^2 } \pr{ \frac{1}{\nh}\sum_{i\in \sample_j} 
\norm{\ti}^2 } \tag{Cauchy--Schwarz} \\ 
&= O_p(1)   \frac{1}{\nh}\sum_{i\in \sample_j} \norm{\estOptj(\zi) -
\Upsilon(\zi)}^2   \tag{Since
$\E \norm{\ti^2} < \infty$ in condition 3} \\ 
& \pto 0
\end{align*}
The last step follows, because the nonnegative random variable $
\nh^{-1}\sum_{i\in \sample_j}
\norm{\estOptj(\zi) -
\limitOpt(\zi)}^2 \ge 0$ has expectation converging to zero by condition 1
of \cref{cond:simple}.
Therefore \[
\Pn \pr{\estOpt(\zi) - \limitOpt(\zi)}\ti^\t = \frac{1}{2} \sum_{j\in
\{1,2\}} \frac{1}{\nh} \sum_{i\in
\sample_j} \pr{\estOptj(\zi) - \limitOpt(\zi)}\ti^\t  = o_p(1). 
\]

We now consider the second statement. Again we may decompose
\[
\Pn \pr{\estOpt(\zi) - \limitOpt(\zi)}\ui = \frac{1}{2} \sum_{j\in
\{1,2\}} \frac{1}{\nh} \sum_{i\in
\sample_j} \pr{\estOptj(\zi) - \limitOpt(\zi)}\ui \equiv \frac{Q_1 + Q_2}{2}
\]
and show that $
\sqrt{\nh} Q_j \defeq \sqrt{\nh} \frac{1}{\nh} \sum_{i\in
\sample_j} \pr{\estOptj(\zi) - \limitOpt(\zi)}\ui \defeq \sqrt{\nh}\frac{1}{\nh}
\sum_{i\in \sample_j} \Delta_i
\ui = o_p(1),$ where we write $\estOptj(\zi) - \limitOpt(\zi) = \Delta_i$
as a shorthand. 

It suffices to show that $
\var\pr{    Q_j } = o(1)$, since $Q_j = \E Q_j + O_p(\sqrt{\var(Q_j)})$ and
$\E[Q_j] = 0$. Note that \begin{align*}
\var(Q_j) = \frac{1}{n} \sum_i \E[\Delta^2_i \cdot \E[U_i^2 \mid S_{-j},
\zi]] \le \frac{1}{n}\sum_i \E[\Delta_i^2] \cdot M \tag{Condition 3,
\cref{cond:simple}}
\end{align*}
which vanishes.
\end{proof}

\begin{lemma}
\label{lemma:variance}
    Under \cref{as:variance} and condition 3 of \cref{cond:simple}, $\hat
    \VS \defeq \Pn
    (\yi - \ti^\t\mlssn)^2
\estOpti\estOpti^\t \pto \VS$.
\end{lemma}
\begin{proof}
    Observe that $(\yi - \ti^\t\mlssn)^2 = \ui^2 + (\theta - \mlssn)^\t 
    (\ti\ti^\t (\theta - \mlssn) + 2\ui\ti) \defeq \ui^2 +  (\theta -
    \mlssn)^\t\vi$, where we define $\vi = \ti\ti^\t (\theta - \mlssn) +
    2\ui\ti$ .
    Note that \begin{align*}
        \norm[\bigg]{(\theta - \mlssn)^\t\Pn \vi \estOpti\estOpti^\t}_F
        &\le  \norm{\theta - \mlssn} \cdot
    \Pn \norm{\vi}
    \norm{\estOpti}^2  \\ &\le o_p(1) \pr{\Pn \vi^2 \cdot \Pn 
    \norm{\estOpti}^4}^{1/2}
    \end{align*}
    $\Pn \norm{\vi}^2$ is $O_p(1)$ if $\norm{\ti}^4, \norm{\ui}^4$ have
    bounded expectations. $\Pn 
    \norm{\estOpti}^4$ is $O_p(1)$ since $\limitOpt$ has bounded fourth
    moments and so does the difference $\norm{\estOpt - \limitOpt}$. Thus
    \[
    \hat \Omega = \Pn \ui^2 \estOpti\estOpti^\t + o_p(1).
    \]
    
    Next, we may compute \[
            \norm[\bigg]{\Pn (\estOpti \estOpti^\t - \limitOpti
            \limitOpti^\t)
    \ui^2}_F \le \Pn 2\ui^2\norm{\limitOpti} \norm{\estOpti
    - \limitOpti} + \Pn \norm{\estOpti
    - \limitOpti}^2 \ui^2.
\]
    Note that the expectation of the second term on the right-hand side
    vanishes:
    \[
    \E[\norm{\estOpti
    - \limitOpti}^2 \E[U_i^2 \mid S_{-j}, \zi]] \le M \E[\norm{\estOpti
    - \limitOpti}^2] \to 0.
    \]
    Thus the second term is a nonnegative sequence with
    vanishing expectation, and is hence $o_p(1)$. To show that the first
    term is $o_p(1)$, it suffices to show that \[
    \E[\ui^2\norm{\limitOpti} \norm{\estOpti
    - \limitOpti}] = o_p(1).
    \]
    This is in turn true since, by condition 3 of \cref{cond:simple} and
    Cauchy--Schwarz \[
    \E[\ui^2\norm{\limitOpti} \norm{\estOpti
    - \limitOpti}] < M \cdot \sqrt{\E[\norm{\limitOpti}^2] \E[\norm{\estOpti
    - \limitOpti}^2]} = o_p(1). \qedhere
    \]
\end{proof}

\begin{theorem}[\cref{thm:weakiv} in the main text]
    Under \cref{as:weakiv}, $\AR_j(\tau_0) \wto \chi^2_{\dim
    \di}$.
\end{theorem}

\begin{proof}
We first show that $\vn \wto \Norm(0,
\Omega)$. Observe that $\tildeui = -\deltadiff^\t \tildexi + \ui$ where
\[\deltadiff =\bk{\frac{1}{\nh}\sum_i
\tildexi \tildexi^\t}^{-1} \frac{1}{\nh} \sum_i \tildexi\ui^\t.\] Then
\begin{align*}
\vn&= \frac{1}{\sqrt{\nh}} \sum_{i\in\sample_j} \estoptj(\zi)\ui - \pr{
\frac{1}{\sqrt{\nh}} \sum_
{i\in\sample_j}\estoptj(\zi)\tildexi^\t} \deltadiff \\
&= \frac{1}{
\sqrt{\nh}}
\sum_{i\in\sample_j} \estoptj(\zi)\ui - \pr{\frac{1}{\nh} \sum_
{i\in\sample_j} \estoptj(\zi)\tildexi^\t}\sqrt{\nh}\deltadiff \\ 
&= \frac{1}{
\sqrt{\nh}} \sum_{i\in\sample_j} \pr{\estoptj(\zi) - \lambda_n^\t \tildexi}
\ui + o_p(1).
\end{align*}
The last equality follows from expanding $\deltadiff$ and applying the
following laws of large numbers
(in triangular arrays):\[
\frac{1}{\nh}\sum_i \tildexi \tildexi^\t = \underbrace{\E[\tildexi
\tildexi]}_{\text{invertible}} + o_p(1) 
\quad \frac{1}{\nh} \sum_i \estoptj(\zi) \tildexi^\t = \E
[\estoptj(\zi) X_i^\t \mid \estoptj] + o_p(1),
\]
for which the fourth-moment conditions (iii), (iv) of
\cref{as:weakiv} are sufficient. The conditions (i) and (ii) are then
Lyapunov conditions for the central limit theorem $\vn \wto \Norm(0,
\Omega)$ conditional on $\estoptj$. Since the limiting distribution does
not depend on $\estoptj$ and the conditions are stated as
$\estoptj$-almost-sure, $\vn \wto \Norm(0,
\Omega)$ unconditionally as well.\footnote{In one dimension, $\P(Z \le t
\mid \estoptj) \asto \Phi(t)$ implies that $\P(Z \le t) \to \Phi(t)$ by
dominated convergence. We may reduce the multidimensional case to the
scalar case with the
Cramer--Wold device.}

Next, we show that $\omegan \pto \Omega$. By condition (ii) and law of
large numbers (so that $\frac{1}{\nh}\sum_i \ui^2 (\estoptj(\zi) - \lambda_n \tildexi) (\estoptj
        (\zi)- \lambda_n \tildexi)^\t \pto \Omega$), it suffices to show
        that \[
\omegan = \frac{1}{\nh}\sum_i \ui^2 (\estoptj(\zi) - \lambda_n \tildexi) (\estoptj
        (\zi)- \lambda_n \tildexi)^\t + o_p(1). 
\]
Write $\tildeui = \ui - \deltadiff^\t \tildexi$ and $\estoptjtilde =
[\estoptj - \lambda_n \tildexi] - (\tilde\lambda - \lambda_n)^\t \tildexi$.
Expanding the sum yields \[
\omegan = \frac{1}{\nh}\sum_i \ui^2[\estoptj - \lambda_n \tildexi][\estoptj
- \lambda_n \tildexi]^\t + \deltadiff^\t \pr{\frac{1}{\nh}\sum_i A_{in}} + 
(\tilde\lambda -
\lambda_n)^\t \pr{\frac{1}{\nh}\sum_i B_{in}}
\]
for some $A_{in}, B_{in}$ that involve products of up to four terms  of $\ui,
\tildexi, \estoptj$. Since the fourth moments are bounded by (iii), we have
that $\frac{1}{\nh}\sum_i A_{in} = O_p(1)$ and $\frac{1}{\nh}\sum_i B_{in}
= O_p(1)$. Since $\tilde \delta - \delta$ and $\tilde \lambda - \lambda_n$
are both
$o_p(1)$, we have the desired expansion.  Therefore, by Slutsky's theorem,
$\AR_j(\tau) \wto Z^\t \Omega^{-1} Z \sim
\chi^2_{\dim \di}$ where $Z \sim \Norm(0,\Omega)$.
\end{proof}

\section{Discussion related to NPIV} 
\label{asec:npiv}
A principled modeling approach is the NPIV model, which treats the unknown
structural function $g$ as an infinite dimensional parameter and considers
the model \[
\E[Y - g(T) \mid Z] = 0. \tag{NPIV}
\]
The researcher may be interested in $g$ itself, or some functionals of $g$,
such as the average derivative $ \theta = 
\E\bk{\diff{g}{D}(T) \mid Z}
$ or the best linear approximation $\beta = \E[TT^\t]^{-1}\E[Tg(T)].$ One might
wonder whether choosing a parametric functional form in place of $g(T)$ is
without loss of generality. Linear regression of $Y$ on $T$, for instance,
yields the best linear approximation to the structural function $\E[Y|T]$,
and thus has an attractive nonparametric interpretation; it may be tempting
to ask whether an analogous property holds for IV regression. If an
analogous property does hold, we may have license in being more blas\'e
about linearity in the second stage.

Unfortunately, modeling $g$ as linear does not produce the best linear
approximation, at least not with respect to the $L^2$-norm.\footnote{In fact,
since it is possible for $g(\cdot)$ to be strictly monotone, instrument $Z$ to
be strong, and $\cov(Y,Z) = 0$, TSLS is not guaranteed to recover any
convex-weighted linear approximation to $g$ either.}
\cite{escanciano2020optimal} show that the best linear approximation can be
written as a particular IV regression estimand\[
\beta = \E[h(Z)T^\t]^{-1}\E[h(Z)Y]
\]
where $h$ has the property that $\E[h(Z) | T] = T$. Note that with
efficient instrument in a homoskedastic, no-covariate linear IV context as
we consider in \cref{sub:simple}, the optimal
instrument is $\hat D(W) = \E[D|W]$. A sufficient condition, under which the
 IV estimand based on the optimal instrument is equal to the best
 linear approximation $\beta$, is the
somewhat strange condition that the projection onto $D$ of \emph{predicted
$D$} is linear in $D$ itself: For some invertible $A$, $
\E[\hat D(W) | D] = AD.
$
The condition would hold, for instance, in a setting where $D,W$ are
jointly Gaussian and all conditional expectations are linear, but it is
difficult to think it holds in general. As such, linear IV would not
recover the best linear approximation to the nonlinear structural function
in general.

A simple calculation can nevertheless characterize the bias of the linear
approach if we take the estimand to be the best linear approximation to the
structural function. Suppose we form an instrumental variable estimator
that converges to an estimand of the form 
\[
\gamma = \E[f(Z) T^\t]^{-1}\E[f(Z)Y].
\]
It is easy to see that \[
\gamma - \beta = \ip{g-\E^*[g|T], \mu-\E^*[\mu | T]},
\]
where $\ip{A, B} = \E[AB]$, $\mu(T) = \E[f(Z)|T]$, and $\E^*[A | B]$ is
the best linear projection of $A$ onto $B$. This means that the two
estimands are identical if and only if at least one of $\mu(\cdot)$ or $g
(\cdot)$ are linear, and all else equal the bias is smaller if $\mu$ or
$g$ is more linear. Importantly, $\mu - \E^*[\mu | T]$ are objects that we
could empirically estimate since they are conditional means, and in
practice the researcher may estimate $\mu - \E^*[\mu | T]$, which delivers
bounds on $\gamma - \beta$ through assumptions on linearity of $g$.

\section{Monte Carlo example}
\label{asec:monte_carlo}

\subsection{Without covariates} 
\label{asub:nocov}

We consider a Monte Carlo experiment where
there are three instruments
$W_0, W_1, W_2 \iid \Norm(0,1)$, one binary treatment variable $D$, and an
outcome $Y$. The
probability of treatment is a nonlinear function of the instruments \[
\P(D=1 \mid W) = \sigma(3\mu(W_0, W_1)) \sin(2W_2)^2 \quad \sigma(t) = 
\frac{1}{1
+ e^{-t}}
\]
where \[\mu(W_0, W_1) = \begin{cases}
  0.1 & W_0^2 + W_1^2 > 1 \\ 
  \sgn(W_0W_1)(W_0^2 + W_1^2) & \text{otherwise}
\end{cases}\]
Naturally, the choice of $\mu$ implies an XOR-function-like pattern in the
propensity of getting treated, where $D=1$ is more likely when $W_0, W_1$
is the same sign, and less likely when $W_0, W_1$ is of different signs. An
empirical illustration of the joint distribution of $D, W_0, W_1$ is in
\cref{fig:xor_pic}. 
\begin{figure}[tb]
  \centering
\includegraphics{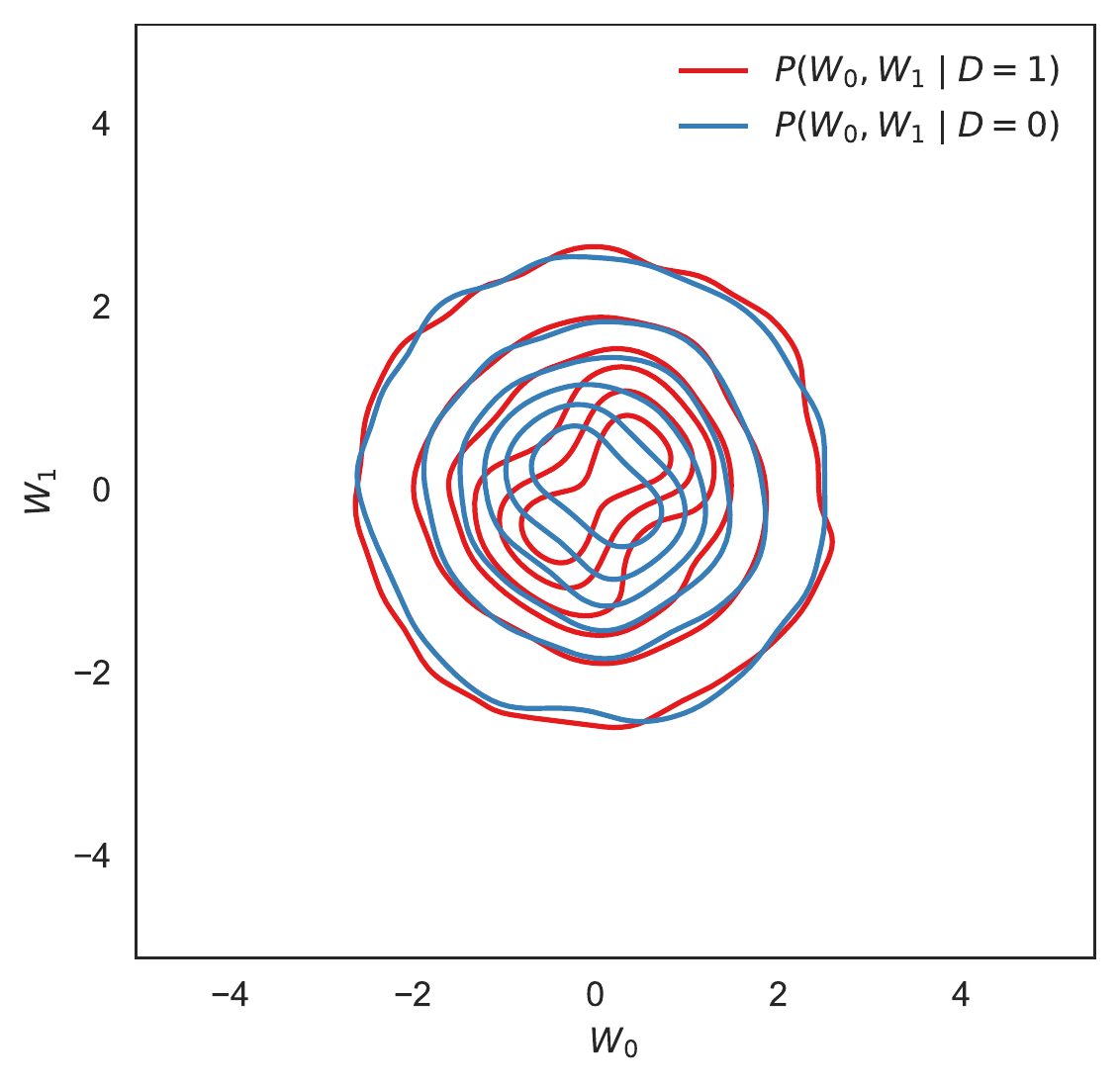}
  \caption{Conditional density of the first two instruments $W_0, W_1$
  given treatment status}
  \label{fig:xor_pic}
\end{figure}
The outcome $Y$ is generated by $Y = D + v(W_0, W_1, W_2) U$, where \[
U = 0.5 (D - \P(D=1 \mid W)) |Z_1| + \sqrt{1-0.5^2} Z_2
\]
where $Z_1, Z_2 \sim \Norm(0,1)$ independently, and $v(W_0, W_1, W_2) =
0.1 + \sigma((W_0 + W_1)W_2).$ Importantly, by construction, $\E[U \mid W]
= 0$, and the true treatment effect is $\tau = 1$.

We consider a variety of estimators for the first stage $\E[D \mid W]$ in a
split-sample IV estimation routine. In particular, we consider two
machine learning estimators\footnote{Of course, ``machine learning
estimators'' is not, strictly speaking, a well-defined distinction.} (LightGBM
and random forest) versus a variety
of more classical linear regression estimators based on taking
transformations (polynomial or discretization) of $W_0, W_1, W_2$ and
estimating via OLS on the transformed instruments. For the traditional,
linear regression-based estimators, we also consider TSLS without sample
splitting. The performances of the estimators, as well as their
definitions, are summarized in \cref{fig:nocov}. 

\begin{figure}[tb]
 { \centering
    \includegraphics[width=\textwidth]{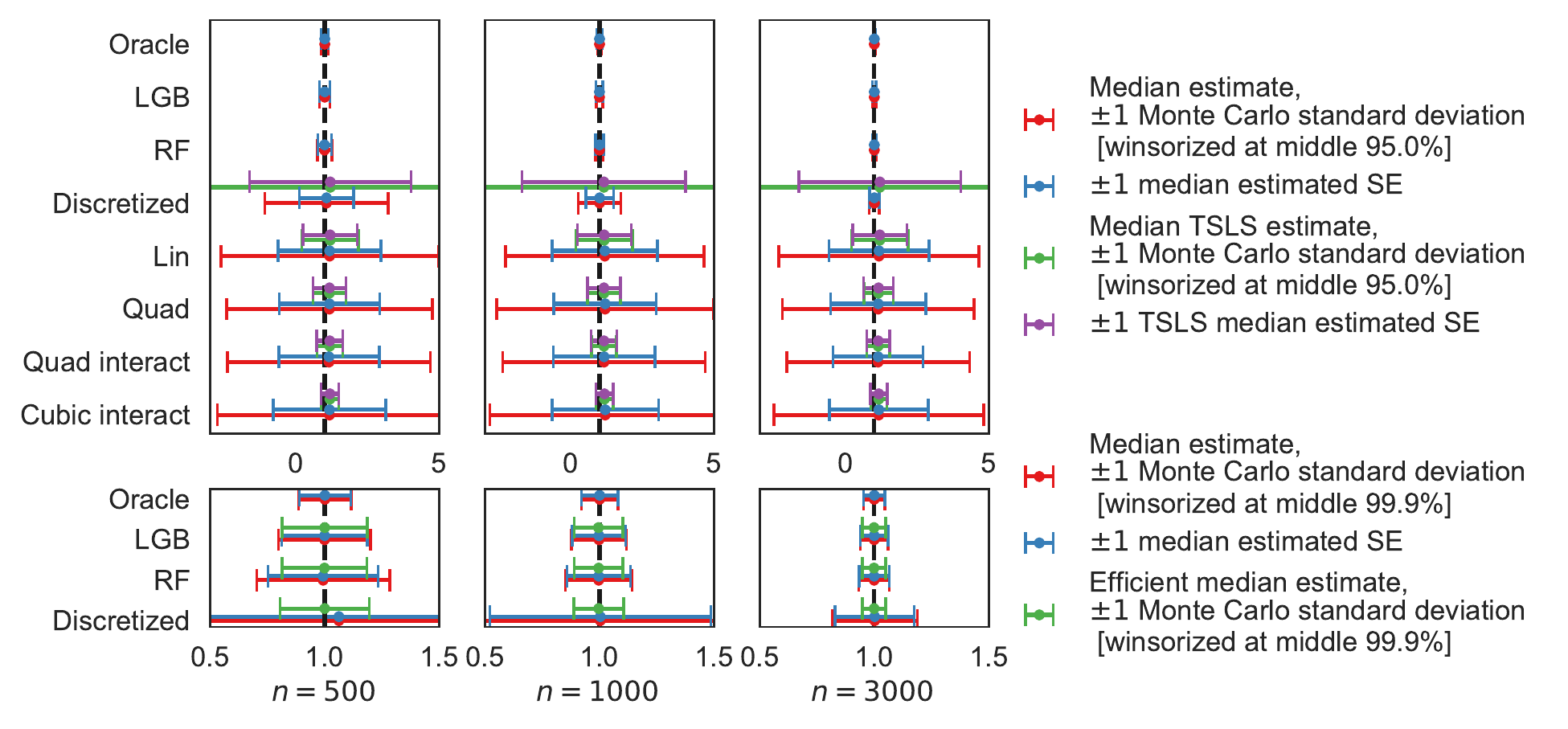}
   } 
  \textbf{Notes:} Median estimates, winsorized standard deviation, and median
  estimated standard error reported for a variety of estimators. The
  Monte Carlo standard deviations are computed from winsorized point
  estimates since the finite-sample variance of the linear IV estimator may
  not exist. 
  
  \quad The
  estimators are split-sample IV estimators that differ in their
  construction of the instrument:   
  \textsf{Oracle} refers to using the true form $\P(D=1 \mid W)$ as the
  instrument; \textsf{LGB} refers to using LightGBM, a gradient boosting
  algorithm; \textsf{RF} refers to using random forest; \textsf{Discretized}
  refers to discretizing $W_0, W_1, W_2$ into four levels at thresholds
  $-1,0,1$, and using all $4^3$ interactions as categorical covariates;
  \textsf{Lin} refers to linear regression with $W_0,W_1,W_2$ untransformed;
  \textsf{Quad} refers to quadratic regression without interactions; 
  \textsf{Quad interact} refers to quadratic regression with full
  interactions; and \textsf{Cubic interact} refers to cubic regression with
  full interactions. 
  
  \quad The latter five estimators (\textsf{Discretized} through \textsf{Cubic
  interact})
  may also be implemented directly with TSLS without
  sample splitting, and we also plot the corresponding performance
  summaries in the top panel (in green and purple). 
  
  \quad The bottom panel is a zoomed-in version of the best performing
  estimators. We also show the performance of the efficient estimator 
  (estimating the optimal instrument with inverse-variance weighting) in
  the bottom panel in green.
  
  \caption{Performance of a variety of estimators for $\tau$ in the setting
  of
  \cref{asub:nocov}}
  \label{fig:nocov}
\end{figure}

\begin{figure}[tb]
  \centering
 \includegraphics{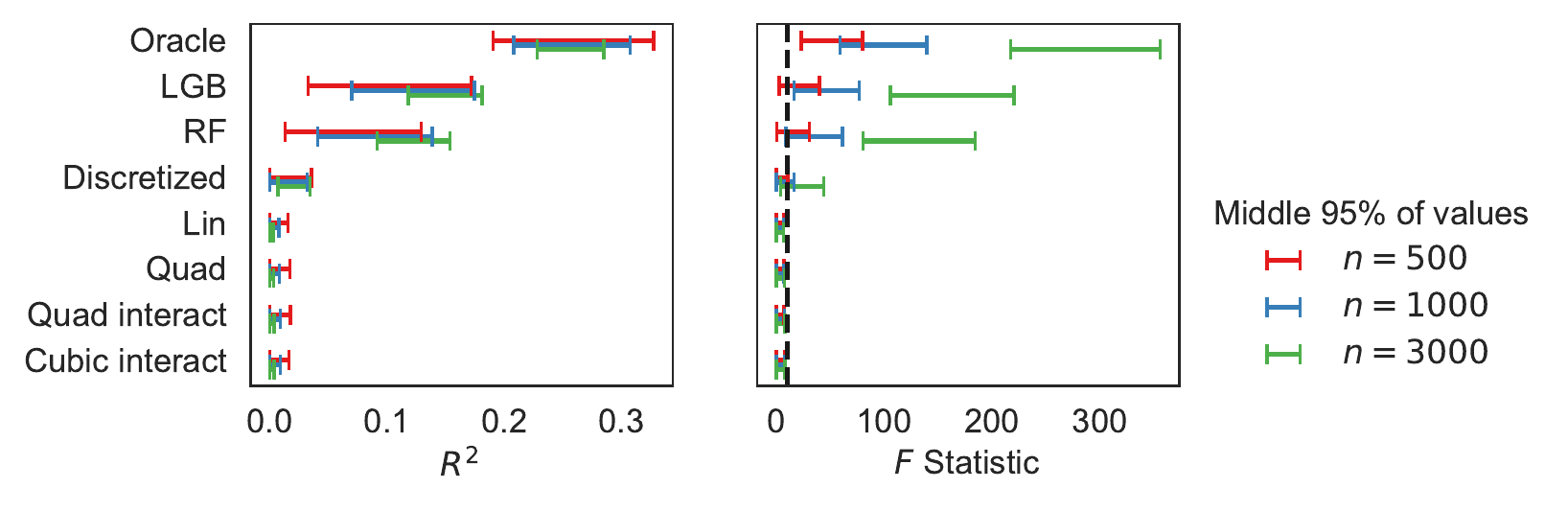}
   \caption{Fit statistics of a variety of estimators for $\tau$ in the
   setting of
  \cref{asub:nocov}. The dashed line reports the \cite{stock2005testing} rule of thumb of
  $F = 10$.}
  \label{fig:nocov_fit}
\end{figure}

\begin{figure}[tb]
  \centering
 \includegraphics[width=\textwidth]{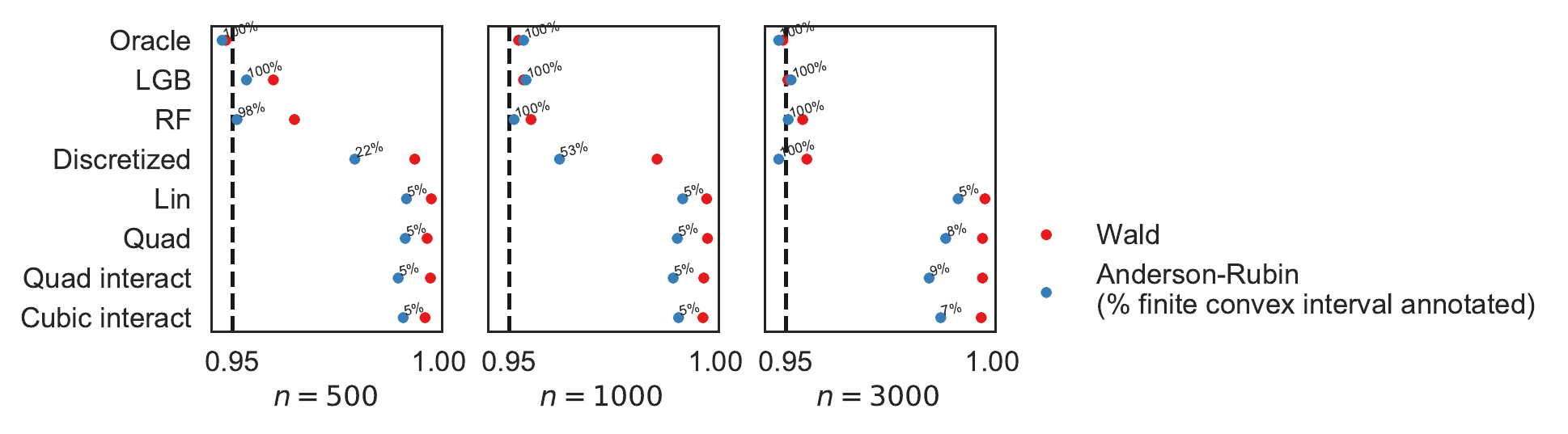}
   \caption{Wald and Anderson--Rubin coverage rates in the setting of
  \cref{asub:nocov}. The parenthesized values are the percentage empirical
  Anderson--Rubin intervals of the finite interval form.}
  \label{fig:nocov_inference}
\end{figure}

We note that flexible estimators appear to be able to discover the complex
nonlinear relationship $\E[D \mid W]$ and produce estimates of $\tau$ that
perform well, whereas more traditional estimators appear to have some trouble
estimating a strong first-stage, resulting in noisy and biased estimates
of $\tau$. In particular, polynomial regression-based
estimators generally have median-biased second stage coefficients with large
variances, particularly if sample-splitting is employed. Among the
linear-regression-based estimators, the discretization-based estimator
appears
to benefit significantly with larger sample sizes and sample-splitting.
Unsurprisingly, the flexible estimators have superior measures of
fit in $R^2$ and the first-stage $F$-statistics \cref{fig:nocov_fit}.

Estimating the optimal instrument with inverse variance weighting appears
to deliver modest benefits in improving the precision of the second-stage
estimator when LightGBM and random forest are used to estimate the
instrument, but the benefit is quite substantial when we consider the
discretization-based estimator. 

We report inference performance in \cref{fig:nocov_inference}. Again, we
see that the flexible methods (``machine learning methods'' and, to a
lesser extent, the discretization estimator) perform well, with both Wald
and Anderson--Rubin intervals covering at close to the nominal level.
Meanwhile, methods that fail to estimate a strong instrument produces
confidence sets that are very conservative in the split-sample setting, and
almost always produces Anderson--Rubin confidence sets which do not take
a finite interval shape. 

\subsection{With covariates}
\label{asub:cov}

We modify the above design by including covariates. Let \[
X = AW + V,
\] 
where $V \sim \Norm(0,I)$ and $A = \begin{bmatrix}
  1 & 0.4 & 0.3 \\ 
  0.5 & 2 & 0.2
\end{bmatrix},$ be two covariates: $X=[X_0, X_1]^\t$. 
If $X_0 > 0$, then we flip $D$ with probability 0.3 to obtain $\tilde D$,
and we let \[
\tilde Y = \tilde D + X^\t \colvecb{2}{0.1}{0.3} + U.
\]
as the modified outcome. As before, $\E[U \mid W] = 0$ and the true effect
is $\tau = 1$. However, we note that the changes in the DGP means that the
setting here and the setting in \cref{asub:nocov} are similar but not
directly comparable---it is not clear whether the setting here is easier
or harder to estimate, compare to the setting in \cref{asub:nocov}.

\begin{figure}[tb]
 { \centering
    \includegraphics[width=\textwidth]{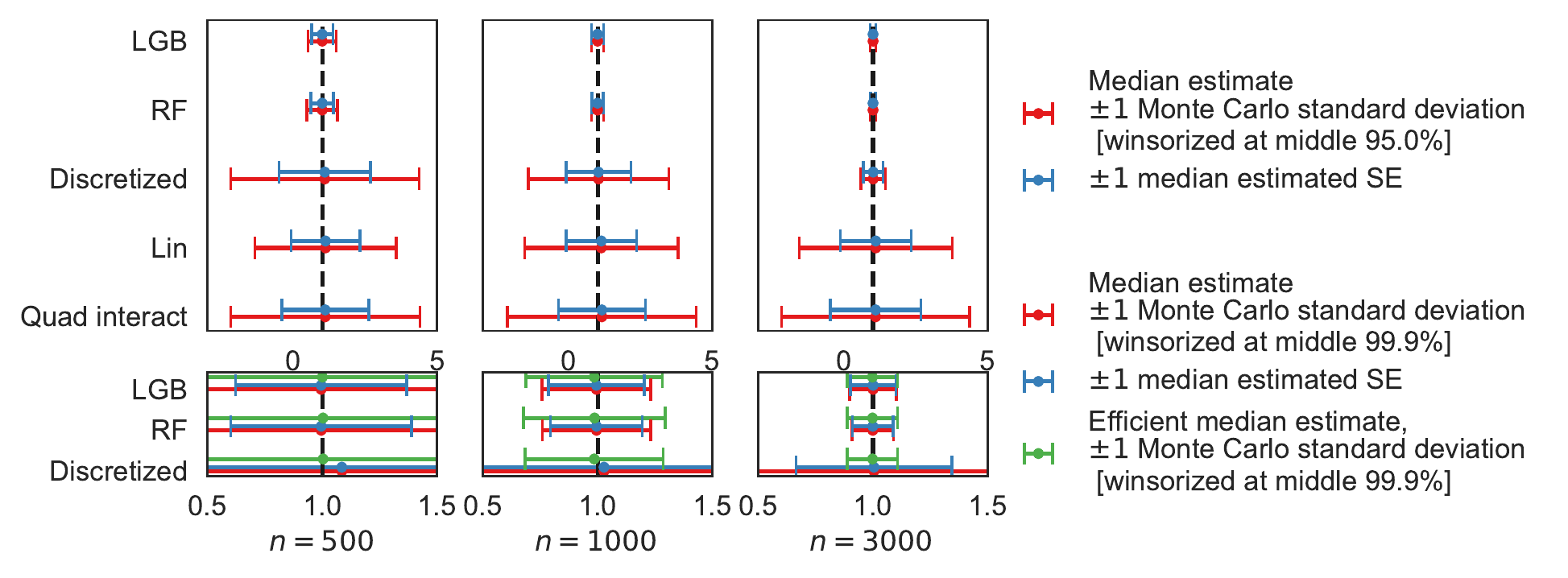}
   } 
  \textbf{Notes:} Median estimates, winsorized standard deviation, and median
  estimated standard error reported for a variety of estimators. The
  Monte Carlo standard deviations are computed from winsorized point
  estimates since the finite-sample variance of the linear IV estimator may
  not exist. 
  
  \quad The
  estimators are split-sample IV estimators that differ in their
  construction of the instrument:   
  \textsf{LGB} refers to using LightGBM, a gradient boosting
  algorithm; \textsf{RF} refers to using random forest; \textsf{Discretized}
  refers to discretizing $W_0, W_1, W_2$ into four levels at thresholds
  $-1,0,1$, and using all $4^3$ interactions as categorical covariates;
  \textsf{Lin} refers to linear regression with $W_0,W_1,W_2$ untransformed;
  \textsf{Quad interact} refers to quadratic regression with full
  interactions. 
  
  \quad The bottom panel is a zoomed-in version of the best performing
  estimators. We also show the performance of the efficient estimator 
  (estimating the optimal instrument with inverse-variance weighting) in
  the bottom panel in green.
  
  \caption{Performance of a variety of estimators for $\tau$ in the setting
  of
  \cref{asub:cov}}
  \label{fig:cov}
\end{figure}

\begin{figure}[tb]
  \centering
 \includegraphics[width=\textwidth]{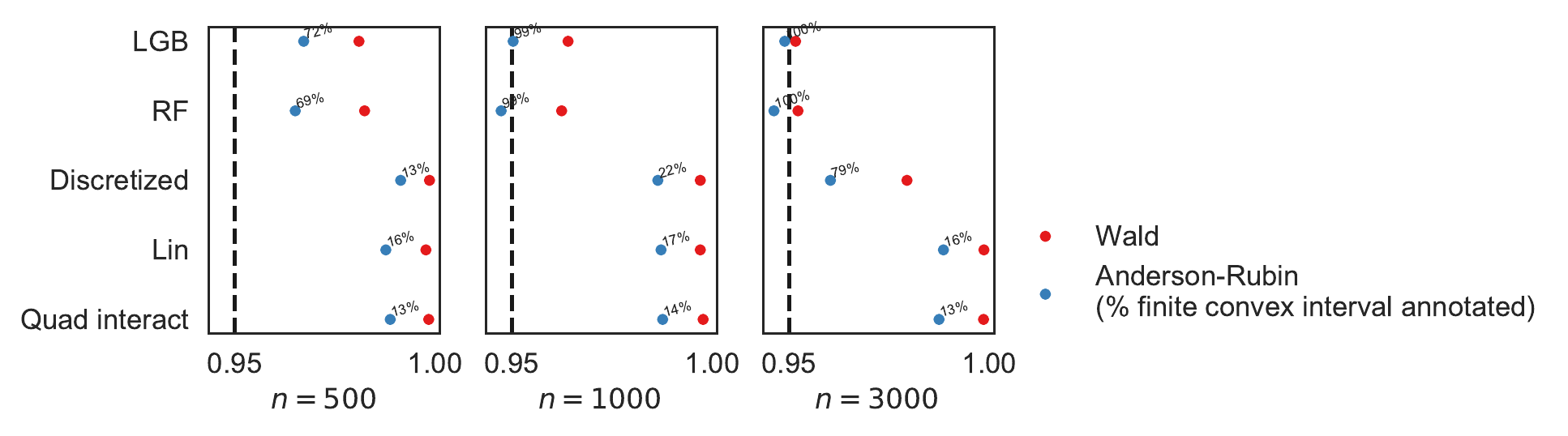}
   \caption{Wald and Anderson--Rubin coverage rates in the setting of
  \cref{asub:cov}. The parenthesized values are the percentage empirical
  Anderson--Rubin intervals of the finite interval form.}
  \label{fig:cov_inference}
\end{figure}

We show the estimation performance in \cref{fig:cov} and inference
performance in \cref{fig:cov_inference}, much like we did in
\cref{fig:nocov,fig:nocov_inference}. Again, we generally have better performance for flexible
methods, and polynomial regression methods do not appear to have comparable
performance. In this case, it appears that when the sample size is small,
even the machine learning based methods sometimes result in an estimated
instrument that is weak. Moreover, interestingly, in this case, efficient
methods do not produce appreciable improvement over methods without inverse
variance weighting, for settings where LGB and RF are used to construct
instruments. If anything, inverse variance weighting seems to do somewhat
worse in finite samples, potentially due to additional noise in the
estimation procedure. 

\section{Interpretation under heterogeneous treatment effects}
\label{asec:hte}

Suppose $\limitOpt(\wi) = a(\wi) [1, b (\wi)]'$ for some
scalar functions $a,b$ with $a(\cdot) \ge 0$ and $\E [a(\wi)] < \infty$.
Then the corresponding linear IV estimand, using $\limitOpt$ as instrument,
can be written as a weighted average of \emph{marginal treatment effects}, a
result due to \cite{heckman2005structural}, which we reproduce here: \[
\tau_\limitOpt = \int_0^1 \weight(\vmte)\cdot  {\MTE
(\vmte)}_{}\,d\vmte \defeq \int_0^1 \weight(\vmte)\cdot  \E\bk{\frac{a(\wi)} {\E a
(\wi)}(Y_1 - Y_0) \mid V=v}\,d\vmte
\]where the weights are \[
\weight(\vmte) \defeq \frac{\E\bk{a(\wi) \tilde b(\wi) \one
(\condmean(\wi) > \vmte)}}{\E[a(\wi) \tilde b(\wi)\condmean(\wi)]}
\quad \tilde b(\wi) \defeq  b(\wi) - \E\bk{\frac{a(\wi)}{\E[a(\wi)]}b
(\wi)}.
\]
The weights $\weight(\cdot)$ integrate to 1 and are \emph{nonnegative} whenever
$b(\wi)$ is a monotone transformation of $\condmean(\wi)$.\footnote{See
Section 4 of \cite{heckman2005structural} for a derivation where $a
(\wi) = 1$. The result with general positive $a(\wi)$ follows with the
change of measure $p(\wi, \di, \yi) \mapsto \frac{a(\wi)}{\E[a(\wi)]}p
(\wi,\di,\yi)$, and so we may simply replace expectation operators
with expectation weighted by $a(\wi)$.}

In the special case where $a(\wi) = 1$ and $b (\wi) = \condmean(\wi)$, which
corresponds to using the optimal instrument under identity weighting, the
estimand is a convex average of marginal treatment effects. In the case
where $a(\wi) = 1/\skedas (\wi)$ and $b(\wi) = \condmean(\wi)$, the estimand
is a convex average of \emph{precision-weighted} marginal treatment effects.
In the heterogeneous treatment effects setting, we stress that efficiency
comparisons are no longer meaningful, since the estimators do not converge
to the same estimand. However, we nonetheless highlight the benefit of using
an optimal instrument-based estimator compared to a standard linear IV
estimator: Optimal instrument-based estimators are guaranteed to recover
convex-weighted average treatment effects, where linear IV estimators with
$\wi$ as the instrument may not.

\end{document}